\documentclass[lettersize,journal]{IEEEtran}
\usepackage{amsmath,amsfonts}
\usepackage{algorithmic}
\usepackage{array}
\usepackage[caption=false,font=normalsize,labelfont=normalsize,textfont=normalsize]{subfig}
\usepackage{amsmath,amssymb,amsfonts}
\usepackage{algorithmic}
\usepackage{graphicx}
\usepackage{multirow}
\usepackage{amsmath}
\usepackage{amssymb}
\usepackage{latexsym}
\usepackage{epsfig}
\usepackage{epsf}
\usepackage{graphics}
\usepackage{listings}
\usepackage{color}
\usepackage{epstopdf}
\usepackage{multirow}
\usepackage{psfrag}
\usepackage{fancyhdr}
\usepackage{multirow}
\usepackage{textcomp}
\usepackage{float}
\usepackage{subfig}
\usepackage{xcolor}
\usepackage{caption}
\usepackage{times}  
\usepackage{helvet}  
\usepackage{courier}  
\usepackage[hyphens]{url}  
\usepackage{graphicx} 
\usepackage{caption} 
\usepackage{cite}
\usepackage{algorithm}
\usepackage{algorithmic}
\usepackage{balance}
\usepackage{xurl}
\usepackage{multicol}
\usepackage{graphicx}
\usepackage{listings}
\usepackage{multicol}
\usepackage[utf8]{inputenc} 
\usepackage[T1]{fontenc}    
\usepackage{url}            
\usepackage{booktabs}       
\usepackage{amsfonts}       
\usepackage{nicefrac}       
\usepackage{microtype}      
\usepackage{xcolor}         
\usepackage{mathtools, amsmath, amsthm}
\usepackage{tikz,pgfplots}
\usepackage{newfloat}
\usepackage{listings}
\usepackage{multicol}

\newtheorem{theorem}{{Theorem}}
\newtheorem{lemma}[theorem]{{Lemma}}

\newtheorem{corollary}[theorem]{{Corollary}}

\newtheorem{definition}{{Definition}}
\newtheorem{remark}{{\textbf{Remark}}}
\usepackage{textcomp}
\usepackage{stfloats}
\usepackage{url}
\usepackage{verbatim}
\usepackage{graphicx}
\newcommand{\deff}{\mbox{$\stackrel{\rm def}{=}$}}
\def\BibTeX{{\rm B\kern-.05em{\sc i\kern-.025em b}\kern-.08em
    T\kern-.1667em\lower.7ex\hbox{E}\kern-.125emX}}
\usepackage{balance}
\begin{document}
\title{Analog Multi-Party Computing: Locally Differential Private Protocols for Collaborative Computations}

\author{Hsuan-Po Liu, Mahdi Soleymani, and Hessam Mahdavifar
\thanks{This paper was presented in part at the IEEE International Symposium on Information Theory (ISIT), Taipei, Taiwan, Jun 2023 \cite{dpcc}. This work was supported in part by Cisco Research under grant 70619503.

Hsuan-Po Liu is with the Department of Electrical Engineering and Computer Science at the University
of Michigan, Ann Arbor, MI 48109, USA (e-mail: hsuanpo@umich.edu). 
Mahdi Soleymani is with the Halıcıo\u{g}lu Data Science Institute at the University of California San Diego, La Jolla, CA 92093, USA (e-mail: msoleymani@ucsd.edu). 
Hessam Mahdavifar is with the Department of Electrical and Computer Engineering at Northeastern University, Boston, MA 02115, and also with the Department of Electrical Engineering and Computer Science at the University
of Michigan, Ann Arbor, MI 48109, USA (e-mail: h.mahdavifar@northeastern.edu).
}
}


\markboth{}%
{How to Use the IEEEtran \LaTeX \ Templates}

\maketitle

\begin{abstract}
We consider a fully-decentralized scenario in which no central trusted entity exists and all clients are \emph{honest-but-curious}. The state-of-the-art approaches to this problem often rely on cryptographic protocols, such as multiparty computation (MPC), that require mapping real-valued data to a discrete alphabet, specifically a finite field. These approaches, however, can result in substantial accuracy losses due to computation overflows. To address this issue, we propose \texttt{A-MPC}, a private analog MPC protocol that performs all computations in the analog domain. We characterize the privacy of individual datasets in terms of $(\epsilon, \delta)$-local differential privacy, where the privacy of a single record in each client's dataset is guaranteed against other participants. In particular, we characterize the required noise variance in the Gaussian mechanism in terms of the required $(\epsilon,\delta)$-local differential privacy parameters by solving an optimization problem. Furthermore, compared with existing decentralized protocols, \texttt{A-MPC} keeps the privacy of individual datasets against the collusion of all other participants, thereby, in a notably significant improvement, increasing the maximum number of colluding clients tolerated in the protocol by a factor of three compared with the state-of-the-art collaborative learning protocols.
Our experiments illustrate that the accuracy of the proposed $(\epsilon,\delta)$-locally differential private logistic regression and linear regression models trained in a fully-decentralized fashion using \texttt{A-MPC} closely follows that of a centralized one performed by a single trusted entity. 
\end{abstract}


\section{Introduction}
\IEEEPARstart{W}{\lowercase{ith}} the intensive increase in demand for distributed computing and learning models over distributed datasets \cite{distml1}, several distributed learning schemes have been proposed that incorporate datasets dispersed among several entities/servers into training models \cite{CodedPrivateML,copml,distml2,distml3,distml4,pmlr-v139-li21e}.
The servers collaborate to jointly train a machine learning model over their individual datasets. One of the major concerns in such distributed systems is to preserve the privacy of the datasets while collaboratively training a model among the servers. Specifically, some categories of datasets may be highly sensitive, e.g., personal medical records, and  \emph{almost} no information should be revealed about the individual records. This motivates proposing multi-party computing (MPC) protocols to keep the privacy of individual records while allowing multiple data owners to collaboratively train a machine learning model without revealing their datasets.

The seminal Shamir’s secret sharing scheme and its various versions are often used to provide information-theoretic security for data, referred to as a secret, while distributing it among a set of servers/users \cite{shamir}. Also, Shamir’s scheme serves as the backbone of most of the existing schemes on privacy-preserving MPC, such as the celebrated BGW scheme \cite{bgw}.
In Shamir’s scheme, the secret/data symbols are always assumed to be elements of a finite field. Consequently, the state-of-the-art schemes treat the data symbols in the given dataset as finite field elements in order to employ Shamir’s secret sharing, see, e.g., \cite{bgw}. However, mapping the data into a finite field can result in substantial accuracy losses, mainly due to computation overflows. Such methods have been considered in several recent works, see, e.g., \cite{CodedPrivateML,copml}, where protocols have been proposed for privately training a logistic regression model in a distributed fashion. 

Recently, a framework has been proposed in \cite{analogss} to construct the counterpart of Shamir’s secret sharing scheme in the analog domain. This framework is then utilized to construct privacy-preserving distributed computation and learning protocols over real/complex datasets. In other words, all the operations, including encoding the data symbols to be distributed among the computational servers and recovery of the final outcome from the collected results returned by the servers are carried out over the infinite fields of $\mathbb{R/C}$. 
The result in \cite[Fig.~4]{analogss} shows that the protocol computed in the analog domain is robust with respect to the size of the training dataset while the fixed-point implementations, i.e., all data symbols are assumed to be elements of a finite field, suffer significantly from wrap-around error as the size of dataset passes a certain threshold depending on the size of the  
underlying finite field.
However, in this approach, the secret cannot be perfectly secured in an information-theoretic sense. In analog domain computations, the information-theoretic measure of security is no longer perfect compared to Shamir’s secret sharing scheme over finite fields. Thus, certain privacy metrics should be analyzed for the protocols in the analog domain to ensure data privacy.
In \cite{dpcc}, a coded computing scheme in the centralized setting has been proposed which considers secret sharing in the analog domain \cite{analogss}. The privacy guarantee in \cite{dpcc} is analyzed by differential privacy \cite{dwork2014algorithmic,10.1145/1866739.1866758,dwork2006our,dwork2006calibrating}.
The protocols in \cite{analogss,dpcc} require a trusted master node to encode data and distribute them among the servers. 

In this paper, we study fully decentralized MPC (i.e., no trusted master node) over real-valued data, guaranteeing local differential privacy [16], [17] for all parties (clients) during the stage that requires data sharing with others. In particular, we consider a distributed setting, where $N$ clients hold their private inputs respectively and engage in a protocol to compute a function of their joint inputs.   One of the distinctions of this paper compared to the previous works considering $(\epsilon,\delta)$-locally differential private MPC schemes in the literature is that we provide a thorough analysis for matrix computations in the domain of $\mathbb{R/C}$, while others establish their analysis for binary datasets. Furthermore, our scheme keeps the privacy of data against a colluding subset of size up to $T$ of clients, which can go up to the maximum possible for the number of adversaries, i.e., $T=N-1$. This is done assuming that all clients are honest-but-curious, i.e., all clients strictly follow the protocol, but they may aggregate their shares to infer information about the data of other clients outside the collusion. 
More specifically, we provide methods for carrying out the required building blocks for computation in the real/complex domain in a privacy-preserving manner. This includes \emph{addition}, \emph{multiply-by-a-constant}, and \emph{multiplication}, where the \emph{multiplication} computation requires two phases of computations, i.e., \emph{offline phase} and \emph{online phase}. In order to carry out multiplication between secret shares, we propose a scheme that leverages \emph{analog multiplication triples} which are the analog counterpart for the Beaver triple \cite{beaver}. The proposed scheme enables us to compute the multiplication between secret shares by performing linear operations without directly multiplying two shares which increases the degree of the polynomial interpolated at the decoder.
Moreover, in theory, if all the computations are done over the real/complex numbers with infinite precision, our protocol can compute the result accurately without error. However, in practice, data is represented by a finite resolution of bits, either as fixed-point or floating-point. Thus, we provide a bound for the perturbation. 

It is worth mentioning that \cite{copml} considers a similar scenario that enables clients to train a logistic regression model collaboratively while no information about the individual datasets or the intermediate model parameters is revealed in an information-theoretic sense. However, guaranteeing perfect privacy imposes a strict upper bound on the maximum number of colluding clients $T$.
In comparison, our protocol tolerates a group of colluding clients up to size $T=N-1$, while keeping the individual datasets and the intermediate model parameters locally differential private. 

We propose two collaborative machine learning algorithms based on the proposed \texttt{A-MPC} satisfying local differential privacy guarantee, which are $(\epsilon,\delta)$-locally differential private logistic regression model for binary classification and $(\epsilon,\delta)$-locally differential private linear regression model. 
The proposed $(\epsilon,\delta)$-locally differential private algorithms are subsequently applied to real datasets to showcase the performance of the proposed \texttt{A-MPC} in practice. 
Our experiments demonstrate that the accuracy of the proposed $(\epsilon,\delta)$-locally differential private algorithms trained in a fully-decentralized fashion using \texttt{A-MPC} closely follows that of a centralized one performed by a single trusted entity with a negligible loss.

The rest of the paper is structured as follows.
In Section~\ref{Preliminary}, we provide some preliminaries. 
In Section~\ref{sec:protocol}, we propose \texttt{A-MPC}. 
In
Section~\ref{sec:LDP}, the privacy guarantees of \texttt{A-MPC} are characterized in terms of local differential privacy measures.
 In Section~\ref{sec:exp},  experimental results over real-world datasets are provided.
 Finally, we conclude the paper in Section~\ref{sec:CCS}. 

\section{Preliminaries and Problem Formulation}
\label{Preliminary}
In this section, we briefly overview secure MPC and secret sharing protocols along with the definitions for local differential privacy that are used later. Also, we formally define the setting considered in this paper.
\subsection{Secure MPC}
\textcolor{black}{Secure MPC \cite{bgw,yao,gmw,smpc4} allows a group of clients to jointly compute public functions on their private inputs assuming that some clients may collude to deduce some information about the private input of other clients. An MPC protocol is considered perfectly secure if the clients can learn only the final computation result while inferring no other information regarding the inputs. Several cryptographic techniques exist for secure MPC, including secret sharing \cite{shamir,bgw}, garbled circuits \cite{gc}, homomorphic encryption \cite{he1,he2,he3,he4}, and oblivious transfer \cite{ot1,ot2}. Such tools are leveraged as building blocks in conventional MPC schemes to prevent the leakage of information about the clients' private inputs. These techniques differ from our approach in various aspects, such as the proportion of corrupted/adversarial clients tolerated 
and/or whether input data belongs to a binary or real field. 
\color{black}The secure MPC can be applied to practical distributed computation tasks, one intensively studied task focused on privacy-preserving machine learning  \cite{secureml,Chameleon,ABY3,securenn,falcon,FLASH}.
However, each protocol is designed specifically for a fixed number of clients, usually no more than four.
In the work, we aim to extend the number of clients in the protocol to an arbitrary number. \color{black} Furthermore, existing prior works on this topic are based on finite field computations with a fixed number of clients.
In this work, we consider secure MPC based on secret sharing in the analog domain as the building block of our proposed protocols. 
Recently, several analog distributed computing protocols have been proposed that aim at recovering an approximation to the computation outcome \cite{ jahani2022berrut, soleymani2022approxifer, jeong2021, jahani2021codedsketch} and/or providing privacy guarantees without mapping real-valued data to the elements of a finite field \cite{tjell2021privacy, makkonen2022secure, soleymani2021analog, makkonen2022analog}. However, all aforementioned schemes require a centralized trusted entity, referred to as \emph{master/fusion} node. One major distinction of the proposed protocols in this work is that they are fully decentralized, i.e., no trusted master node is required to carry out the computations privately.  }

In general, there are two main threat models considered in the secure MPC literature. The \emph{semi-honest adversary} model considers the case where adversarial clients follow the computation protocol, but might collude to infer the secret (i.e., the data of other users) by aggregating their shares.  
In another threat model, there are possibly \emph{malicious adversaries} who may decide not to follow the steps in the protocol in order to corrupt the outcome by sharing incorrect information throughout the protocol. In this paper, we consider the semi-honest adversary threat model. 
The semi-honest adversary is also known as \emph{passive} adversary since they cannot take any actions other than collecting all the information they gathered.

\subsection{Secret Sharing}
Shamir's secret sharing \cite{shamir} is a fundamental building block for various MPC protocols. In this scheme, a secret/dataset is encoded into $N$ secret shares, where $N$ is the number of clients, and then each share is given to one of the clients. 
\textcolor{black}{Secret sharing has been widely employed in cloud-based scenarios to improve the security of sensitive data for clients \cite{cloud1,cloud2,cloud3,cloud4}.}
In the original Shamir’s scheme, the secret/data symbols as well as the operation involving them are done over a finite field. Hence, a common approach for deploying it in practical settings involving real-valued datasets is through quantization and mapping to finite fields which could cause accuracy loss due to quantization errors and computation overflows. 
Thus, to address these critical issues, \cite{analogss} proposed an analog secret sharing scheme that can be directly applied to real-valued data. This approach is also utilized for a distributed computing protocol, where a master node offloads a computational job to a set of workers/clients. To this end, the master encodes a real-valued dataset via analog secret sharing before sharing it with the clients in order to provide privacy guarantees against any set of colluding parties up to a specific size $T$. 
In this paper, we extend such protocols by incorporating significant improvements across various crucial aspects in order to arrive at fully decentralized and privacy-preserving protocols that work over real-valued datasets in a scalable fashion.

\subsection{Differential Privacy}


Differential privacy \cite{dwork2014algorithmic,10.1145/1866739.1866758,dwork2006our,dwork2006calibrating} has received considerable attention as a formal mathematical notion of privacy that provides protection against strong adversaries. To protect the single individual’s dataset, local differential privacy has been discussed \cite{Duchi,hitter}.
The MPC schemes considering (local) differential privacy have been studied in \cite{dphow, mpcdpclassifiers,smpcdp,dperm,ldpevolvingdata}. In particular, \cite {dphow} characterizes the definition of differential privacy for a setup with semi-honest parties. 
 In \cite {mpcdpclassifiers}, authors employ asymmetric key additive homomorphic encryption to compose a perturbed aggregate classifier satisfying differential privacy from classifiers locally trained by multiple untrusted parties. Also, \cite {smpcdp} considers MPC under differential privacy, where each party possesses a single bit of information and the bits are independent. In its proposed algorithm therein, each party broadcasts a randomized version of its bit with certain probabilities by adding random noises to data. 
In \cite{dperm}, output perturbation and gradient perturbation are proposed in a distributed learning setting to ensure the privacy of data while incorporating distributed datasets to train a global model. Each iteration of the training process requires computing securely by transforming the data symbol into the discrete domain, then adding noises to the computation results in order to satisfy the differential privacy requirement. Moreover, \cite{ldpevolvingdata} introduces a locally differential private technique for collecting statistical information from users by utilizing a randomized response scheme. However, all existing works merely consider the case where all parties hold a single bit or a scalar value and all the computations are performed over data after they have been mapped to a finite field. Next, we provide a formal definition of local differential privacy and the framework for analyzing the privacy loss.

\subsubsection{Local Differential Privacy}
The main idea to achieve differential privacy is through perturbation by introducing random noises generated according to a chosen distribution. Local differential privacy considers algorithms to keep each individual user's dataset private. We formally define the notion of $(\epsilon,\delta)$-local differential privacy in the following
\begin{definition}[$(\epsilon,\delta)$-local differential privacy]
\label{def:LDP}
Let $d$ and $d^\prime$ be two neighboring datasets, where $d,d^\prime\in \mathcal{D}$, in an individual client that only differs by a single record, i.e., $\mathrm{dist}(d,d^\prime)=1$. The neighboring datasets $d$ and $d^\prime$ satisfy $(\epsilon,\delta)$-local differential privacy for any $\epsilon>0$ and $\delta \in [0,1]$ under a randomized mechanism $\mathcal{M}$ that under any event $\mathcal{E} \subseteq \mathrm{Range}(\mathcal{M})$,
\begin{equation}
\label{eq:LDP}
\mathbb{P}[\mathcal{M}(d)\in \mathcal{E}]\leq e^\epsilon\cdot\mathbb{P}[\mathcal{M}(d^\prime)\in \mathcal{E}]+\delta,
\end{equation}
where $\delta$ represents the failure probability.
\end{definition}
The sensitivity for a query function $f(\cdot)$ is the largest difference between the actual and the perturbed output.
\begin{definition}[$l_2$ local sensitivity]
\label{eq:ls}
For two neighboring datasets $d$ and $d^\prime$ in an individual client together with a query function $f:\mathcal{D}\to\mathbb{R}$, the $l_2$ local sensitivity is defined as follows:
\begin{equation}
\Delta \overset{\mathrm{def}}{=} \max_{\mathrm{dist}(d,d^\prime)=1}||f(d)-f(d^\prime)||_2.
\end{equation}
\end{definition}
The Gaussian mechanism is also defined as follows.
\begin{definition}[Gaussian mechanism]
\label{eq:gm}
Consider a query function $f$ to be applied on a dataset $d$. Then the Gaussian mechanism $\mathcal{M}$ is defined as 
\begin{equation*}
\mathcal{M}(d) \overset{\mathrm{def}}{=} f(d)+\mathcal{N}(0,\sigma^2),
\end{equation*}
which adds random noise to the query result according to a zero-mean Gaussian distribution with variance  $\sigma^2$.  
\end{definition}
\subsubsection{Analysis of Privacy Loss}
Consider two neighboring datasets $d$ and $d^\prime$ and a query function $f$. 
For a randomized mechanism $\mathcal{M}$, the probability density function (PDF) corresponding to the datasets $d$ and $d'$ are denoted as  $p_{\mathcal{M}(d)}(y)$ and $p_{\mathcal{M}(d^\prime)}(y)$, respectively. 
Let
\begin{equation}
\label{eq:plf}
l_{\mathcal{M},d,d^\prime}(y)\overset{\mathrm{def}}{=}\mathrm{ln}[\frac{p_{\mathcal{M}(d)}(y)}{p_{\mathcal{M}(d^\prime)}(y)}],
\end{equation}
which is referred to as the privacy loss function, and, 
$$L_{\mathcal{M},d,d^\prime}=l_{\mathcal{M},d,d^\prime}(Y),$$
that is referred to as the privacy loss random variable.
Then, the $(\epsilon,\delta)$-local differential privacy implies
\begin{equation}
\label{eq:plrv}
\mathbb{P}[L_{\mathcal{M},d,d^\prime}\leq\epsilon]\geq 1-\delta.
\end{equation}
Consider the case where $f(d)=0$ and $f(d^\prime)=\Delta$, then we have
\begin{equation}
p_{\mathcal{M}(d)}(y)=\frac{1}{\sqrt{2\pi\sigma^2}}e^{-\frac{y^2}{2\sigma^2}},
\end{equation}
and
\begin{equation}
p_{\mathcal{M}(d^\prime)}(y)=\frac{1}{\sqrt{2\pi\sigma^2}}e^{-\frac{(y-\Delta)^2}{2\sigma^2}}.
\end{equation}
Note that although the ratio of probabilities is always positive, the result after taking the logarithm may become negative. Thus, typically, the absolute value of the privacy loss function is considered as
\begin{equation}
\label{eq:plf_gaussian}
|l_{\mathcal{M},d,d^\prime}(y)|
=|\textstyle\frac{1}{2\sigma^2}(2y\Delta-\Delta^2)|.
\end{equation}

Note that by the definition of $(\epsilon,\delta)$-local differential privacy in \eqref{eq:LDP}, one needs 
\begin{equation}
    |\frac{1}{2\sigma^2}(2y\Delta-\Delta^2)|<\epsilon
\end{equation}
with probability at least $1-\delta$ to guarantee $(\epsilon,\delta)$-local differential privacy.

\subsection{Problem Setting}\label{setting}
\begin{figure}[t]
\centering
{\includegraphics[width=.45\textwidth]{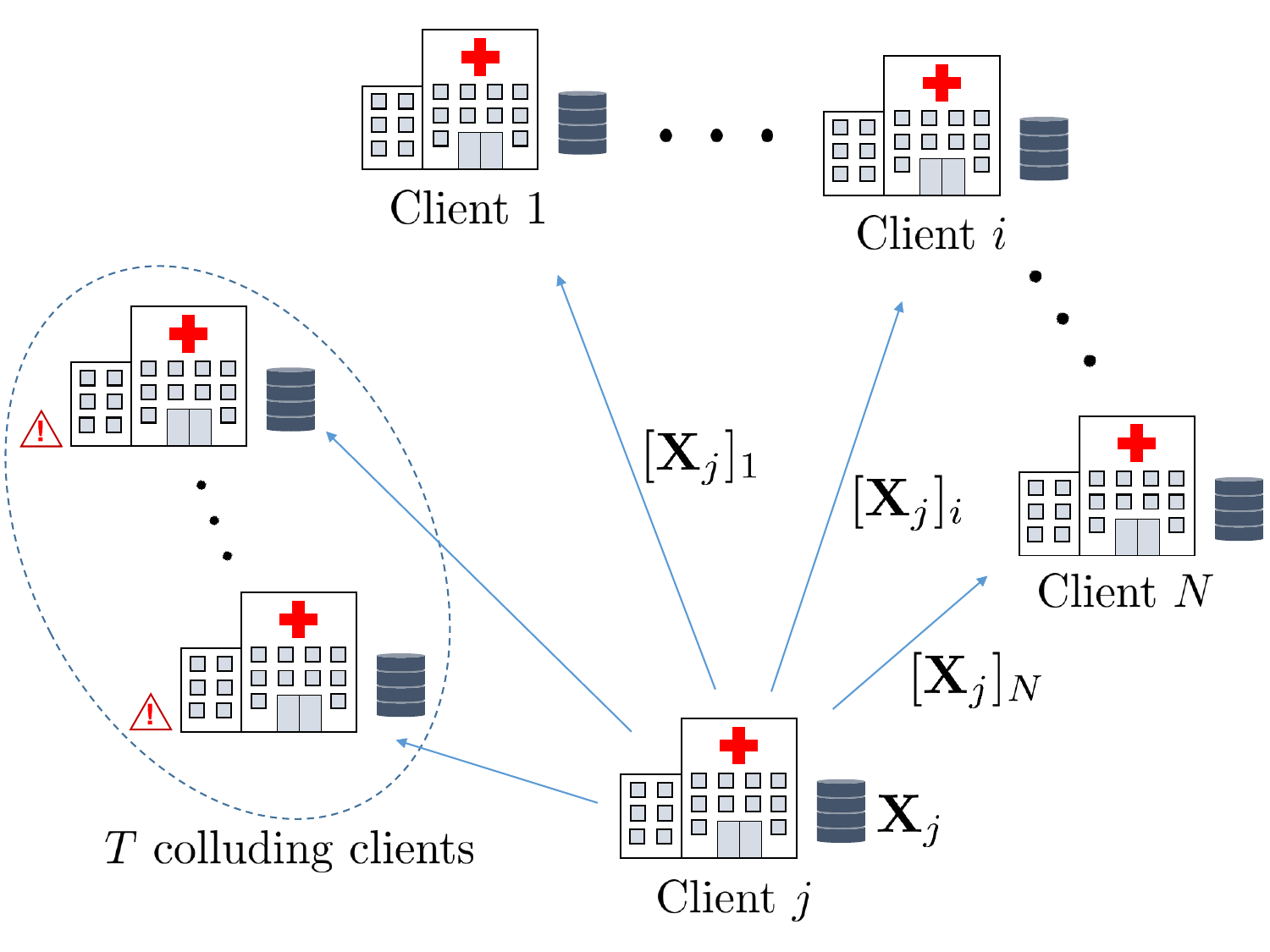}}\\
  \caption{Problem setting}\label{fig:problem_setting}
\end{figure}
In this paper, we consider a decentralized setting where $N$ clients each maintain private datasets,  aiming to execute a predefined protocol for computing a function over the collection of their datasets while ensuring privacy. The setting is as shown in Fig.~\ref{fig:problem_setting}. Let the dataset at the $j$-th client be denoted by $\mathbf{X}_j$, for $j\in[N]$, also referred to as its \emph{secret}. We denote the secret held by client $j$ as $\mathbf{X}_j$ that is shared to the client $i$ as $[\mathbf{X}_j]_i$, for $i,j\in[N]$. The protocol is executed in a synchronous environment with point-to-point secure communication channels between clients and is required to preserve privacy against the collusion of up to $T$ clients, for some $T < N$. All clients are honest-but-curious, which means they strictly follow the protocol. However, colluding clients attempt to infer information about the inputs of the other clients. We require that the local differential privacy be preserved for any collusion of size up to $T$. However, collusion of $T+1$ or more clients may fully reveal the secrets. Note that for the secret sharing stage, the clients simply add noises to their datasets. We solve an optimization problem for choosing the optimal noise parameter to satisfy the $(\epsilon,\delta)$-local differential privacy.

\section{The Proposed \texttt{A-MPC}}
\label{sec:protocol}
In this section, we present \texttt{A-MPC} protocol, a fully decentralized MPC protocol that guarantees the privacy of data under local differential privacy. First, the analog secret sharing for the protocol is illustrated. Then, we discuss linear computations including \emph{addition} and \emph{multiply-by-a-constant}. {Furthermore, a protocol that carries out multiplication between two secret shares in the analog domain is proposed that entails linear computations.}
\subsection{Secret Sharing in \texttt{A-MPC}}
Consider a decentralized system with $N$ clients. Now, we characterize the secret share from client $j$ to client $i$, as $[\mathbf{X}_j]_i$, for $i,j\in[N]$.
Client $j$ randomly generates a polynomial $\mathcal{S}_j(s)$ for sharing the secret $\mathbf{X}_j\in\mathbb{R}^{m\times n}$ with client $i$,
the polynomial follows $\mathcal{S}_j(0)=\mathbf{X}_j$ and $\mathrm{deg}(\mathcal{S}_j(s))=T$, such that
\begin{equation}\label{poly_def}
\mathcal{S}_j(s)=\mathbf{X}_j+\sum_{k=1}^T s^k \mathbf{N}_{j,k},
\end{equation}
where all entries in $\mathbf{N}_{j,k}$'s are noises sampled from \emph{i.i.d.} Gaussian distribution $\mathcal{N}(0,\sigma_s^2)$ by client $j$, for $i,j\in [N]$.
We resample $\sum_{k=1}^T s^k \mathbf{N}_{j,k}$ by randomly generating $\mathbf{N}_{j,k}$'s, until all entries in both $\mathrm{Real}(\sum_{k=1}^T s^k \mathbf{N}_{j,k})$, as the real part of $\sum_{k=1}^T s^k \mathbf{N}_{j,k}$, and $\mathrm{Imag}(\sum_{k=1}^T s^k \mathbf{N}_{j,k})$, as the imaginary part of $\sum_{k=1}^T s^k \mathbf{N}_{j,k}$, are within the range of $[-t,t]$, for $t\in\mathbb{R}^+$ and $i,j\in[N]$, for practical use.

The share sent from client $j$ to client $i$ is the evaluation of the polynomial defined in \eqref{poly_def} at the evaluation point $\omega_i=\exp(\sqrt{-1}\frac{2\pi i}{N})$, i.e., $\mathcal{S}_j(\omega_i)$, where $\omega_i$'s are publicly known parameters, for all $i\in[N]$.
Thus, we have  $[\mathbf{X}_j]_i=\mathcal{S}_j(\omega_i)$, for $i,j\in[N]$.
One can write the secret share sent from client $j$ to client $i$ as follows:
\begin{equation}
\label{eq:Sj}
[\mathbf{X}_j]_i=\mathcal{S}_j(\omega_i)
=  \mathbf{X}_j+\sum_{k=1}^T \omega_i^k\mathbf{N}_{j,k}=\mathbf{X}_j+\tilde{\mathbf{N}}_{ji},
\end{equation}
where $\tilde{\mathbf{N}}_{ji}=\sum_{k=1}^T \omega_i^k\mathbf{N}_{j,k}$, for $i,j\in[N]$.
We resample $\sum_{k=1}^T \omega_i^k \mathbf{N}_{j,k}$ by randomly generating $\mathbf{N}_{j,k}$'s, until all entries in both $\mathrm{Real}(\tilde{\mathbf{N}}_{ji})$, as the real part of $\tilde{\mathbf{N}}_{ji}$, and $\mathrm{Imag}(\tilde{\mathbf{N}}_{ji})$, as the imaginary part of $\tilde{\mathbf{N}}_{ji}$, are within the range of $[-t,t]$, for $t\in\mathbb{R}^+$ and $i,j\in[N]$, for practical use. The truncated Gaussian distribution with zero mean and a resampling parameter $t$ is denoted by $\mathcal{TN}(0,\sigma^2;[-t,t])$ and its PDF is given as
\begin{equation}
\label{eq:trunpdf}
\begin{aligned}
p_{\tilde{{N}}}(y)
=\frac{\phi(y)}{2\Phi(\frac{t}{\sigma})-1}\cdot\mathbb{I}_{[-t,t]}(y),
\end{aligned}
\end{equation}
where $\phi(y)$ is the PDF for $\mathcal{N}(0,\sigma_s^2)$, $\Phi(\cdot)$ is the cumulative density function for the standard normal distribution, and $\mathbb{I}_{[-t,t]}(y)$ is an indicator function such that $\mathbb{I}_{[-t,t]}(y)=1$ for $y \in [-t,t]$, and $\mathbb{I}_{[-t,t]}(y)=0$ otherwise. 

In the following lemma, we show that the combined noises of all entries in $\tilde{\mathbf{N}}_{ji}$ are distributed according to $\mathcal{TN}(0,\sigma^2;[-t,t])$, for $i,j\in[N]$. 
\begin{lemma}
The distribution of all entries in $\tilde{\mathbf{N}}_{ji}$'s is $\mathcal{TN}(0,\sigma^2;[-t,t])$, for $i,j\in[N]$.
\end{lemma}
\begin{proof}
Recall that $\lbrack \mathbf{X}_j\rbrack_i=\mathbf{X}_j+\tilde{\mathbf{N}}_{ji}$, where  $\tilde{\mathbf{N}}_{ji}=\sum_{k=1}^T \omega_i^k\mathbf{N}_{j,k}$ with all entries in $\mathbf{N}_{j,k}$'s are generated independently according to  $\mathcal{N}(0,\sigma^2_s)$, for $i,j\in[N]$. Since the sum of independent Gaussian random variables is also a Gaussian random variable, each entry in the weighted sum $\sum_{k=1}^T \omega_i^k\mathbf{N}_{j,k}$ is distributed according to $\mathcal{N}(0,\sum_{k=1}^T \lvert\omega_i^k\rvert^2\sigma_s^2)$, for $i,j\in[N]$.  
We have
\begin{equation}
\sum_{k=1}^T \lvert\omega_i^k\rvert^2\sigma_s^2=(\sum_{k=1}^T \lvert\omega_i^k\rvert^2)\cdot\sigma_s^2=\sigma^2,
\end{equation}
where $\sigma=\frac{\alpha^*\cdot\Delta}{\sqrt{2\epsilon}}$, for $i\in[N]$.
Note that we truncate the distribution to $[-t,t]$, which results in the distribution $\mathcal{TN}(0,\sigma^2;[-t,t])$, for $i\in[N]$.
\end{proof}
\subsection{The proposed analog computations}
We consider a protocol in the presence of a semi-honest adversary who controls at most $T=N-1$ clients that may collude. Note that this corresponds to the maximum value for $T$ since  $T=N$ corresponds to the trivial case where all individual datasets could be recovered using polynomial interpolation. We show that our protocols preserve the privacy of the individual datasets and yet are capable of recovering the computation outcome in such an extreme worst-case scenario. 
Each client posses an input and output dataset, i.e., $(\mathbf{X}_1,\dots,\mathbf{X}_N)\to(\mathbf{Y}_1,\dots,\mathbf{Y}_N)$, where $\mathbf{Y}_j$ is the desired computation result for the $j$-th client,  for $j\in[N]$.
In order to carry out any polynomial computation over the secret shares at the clients, we need to know how to do the following basic operations: \emph{addition}, \emph{multiply-by-a-constant}, and \emph{multiplication} computations.
 in detail in the following. 
 \subsubsection{Addition computation}
We assume an analog addition for two matrices, $\mathbf{U}_1$ and $\mathbf{U}_2$, where $\mathbf{U}_1,\mathbf{U}_2\in\mathbb{R}^{m \times n}$. In the \texttt{A-MPC} protocol, the $i$-th client holds secret shares for $\mathbf{U}_1$ and $\mathbf{U}_2$ as $[\mathbf{U}_1]_i$ and $[\mathbf{U}_2]_i$, respectively, for $i \in [N]$. The secret shares are
\begin{equation}
\label{eq:addition}
\begin{aligned}
[\mathbf{U}_j]_i
=  \mathbf{U}_j+\sum_{k=1}^T \omega_i^k\mathbf{N}_{j,k}=\mathbf{U}_j+\tilde{\mathbf{N}}_{ji},
\end{aligned}
\end{equation}
where $i\in[N]$ and $j\in\{1,2\}$. Their addition is
\begin{equation}
\label{eq:addition}
[\mathbf{U}_1]_i+[\mathbf{U}_2]_i=\mathbf{U}_1+\mathbf{U}_2+\sum_{k=1}^T \omega_i^k(\mathbf{N}_{1,k}+\mathbf{N}_{2,k}),
\end{equation}
for $i\in[N]$. Thus, each client locally adds its own shares $[\mathbf{U}_1]_i$ and $[\mathbf{U}_2]_i$ together, which results in valid shares of the sum of the inputs held at the clients, as required. One can see that the constant term in \eqref{eq:addition} remains $\mathbf{U}_1+\mathbf{U}_2$.
\subsubsection{Multiply-by-a-constant computation}
Given a matrix $\mathbf{U}\in\mathbb{R}^{m \times n}$, we consider the multiplication of a constant $l$ and $\mathbf{U}$. In the \texttt{A-MPC} protocol, the $i$-th client holds secret share for $\mathbf{U}$ as $[\mathbf{U}]_i$, for $i \in [N]$. Multiplying $[\mathbf{U}]_i$ by the constant $l$ yields:
\begin{equation}
\label{eq:multiplyconst}
\begin{aligned}
l\cdot[\mathbf{U}]_i=l\cdot\mathbf{U}+l\cdot\sum_{k=1}^T \omega_i^k\mathbf{N}_{k}=l\cdot\mathbf{U}+\sum_{k=1}^T \omega_i^k(l\cdot\mathbf{N}_{k}),
\end{aligned}
\end{equation}
for $i\in[N]$.
Thus, each client simply locally multiplies its share $[\mathbf{U}]_i$ with the constant $l$, which results in a valid share of the multiplication held at the clients, as required, for $i\in[N]$. The constant term in \eqref{eq:multiplyconst} remains $l\cdot\mathbf{U}$. 

\subsubsection{Multiplication computation}
The multiplication triple is an efficient method of performing multiplications between secret shares. \textcolor{black}{It is known to reduce communication overhead between the clients by dividing the process into an \emph{offline phase} and an \emph{online phase} as in \cite{beaver}, where the offline phase is an input-independent setup phase that greatly reduces the overhead of the online phase where clients provide their inputs.} We assume an analog multiplication for two matrices, $\mathbf{U}$ and $\mathbf{V}$, \color{black}  where $\mathbf{U}\in \mathbb{R}^{m_1 \times n_1}$ and $\mathbf{V}\in \mathbb{R}^{m_2 \times n_2}$. \color{black}Note that we must have $n_1=m_2$ for the matrix multiplication. In an analog MPC protocol, the $i$-th client holds secret shares for $\mathbf{U}$ and $\mathbf{V}$ as $[\mathbf{U}]_i$ and $[\mathbf{V}]_i$, for $i \in [N]$.
\begin{itemize}
\item \emph{Offline phase}: In the offline phase, the protocol first generates random triplets $\mathbf{A}, \mathbf{B}, \mathbf{C}$ such that $\mathbf{AB}=\mathbf{C}$ where $\mathbf{A}\in \mathbb{R}^{m_1\times n_1}$, $\mathbf{B}\in \mathbb{R}^{m_2\times n_2}$, and $\mathbf{C}\in \mathbb{R}^{m_1\times n_2}$, where  $n_1=m_2$.
All entries in both $\mathbf{A}$ and $\mathbf{B}$ follow the distribution $\mathcal{TN}(0,\sigma^2;[-t,t])$.
Then, the triplets $\mathbf{A}$, $\mathbf{B}$, and $\mathbf{C}$ are secretly shared with all clients by utilizing the analog secret sharing scheme, in such a way that the $i$-th client holds secret shares from each $\mathbf{A}$, $\mathbf{B}$, and $\mathbf{C}$, denoted by $[\mathbf{A}]_i$, $[\mathbf{B}]_i$, and $[\mathbf{C}]_i$, respectively, for $i \in [N]$. Note that the offline phase procedure can be conducted independently of input data and, as such, can be completed before the subsequent online phase, which will be described next.  

\item \emph{Online phase}: At the beginning of this phase, the protocol receives two input data as matrices $\mathbf{U}$ and $\mathbf{V}$. The matrices are then secretly shared with all clients by invoking the analog secret sharing scheme, thus the $i$-th client holds secret shares $[\mathbf{U}]_i$ and $[\mathbf{V}]_i$, for $i\in [N]$. Together with the secret shares from the former phase, the $i$-th client holds secret shares $[\mathbf{A}]_i$, $[\mathbf{B}]_i$, $[\mathbf{C}]_i$, $[\mathbf{U}]_i$, and $[\mathbf{V}]_i$, for $i\in[N]$.

Next, client $i$ computes
$[\mathbf{D}]_i$ and $[\mathbf{E}]_i$ as $[\mathbf{D}]_i=[\mathbf{U}]_i-[\mathbf{A}]_i$ and $[\mathbf{E}]_i=[\mathbf{V}]_i-[\mathbf{B}]_i$, where $\mathbf{D} \in \mathbb{R}^{m_1\times n_1}$ and $\mathbf{E} \in \mathbb{R}^{m_2\times n_2}$, respectively, for for all $i\in[N]$. 
The parties then proceed by collaboratively recovering  $\mathbf{D}$ and $\mathbf{E}$ using the shares $[\mathbf{D}]_i$ and $[\mathbf{E}]_i$. Note that the decoding procedure for $\mathbf{D}$ and $\mathbf{E}$ involves interpolating polynomials that pass through $[\mathbf{D}]_i$ and $[\mathbf{E}]_i$, respectively. This procedure only requires the previously discussed \emph{addition} and \emph{multiply-by-a-constant} subroutines as the evaluation points $\omega_i$'s are public.
Client $i$ then computes its secret share of the multiplication, denoted by $[\mathbf{UV]}_i$, as follows:
\begin{equation}
\label{eq:triples}
[\mathbf{UV]}_i=\mathbf{D}[\mathbf{B}]_i+[\mathbf{A}]_i\mathbf{E}+\mathbf{DE}+[\mathbf{C}]_i,
\end{equation} 
for all $i\in[N]$.

Note that \eqref{eq:triples}  can be implemented by employing both \emph{addition} and \emph{multiply-by-a-constant} computations.
In order to decode the desired computation result $\mathbf{UV}$, each client collects at least $T+1$ secret shares from $\{[\mathbf{UV]}_i\}_{i\in[N]}$. 
Then, the multiplication of the input matrices is recovered by computing $\mathbf{D}\mathbf{B}+\mathbf{A}\mathbf{E}+\mathbf{DE}+\mathbf{C}$. By substituting $\mathbf{D}=\mathbf{U}-\mathbf{A}$ and $\mathbf{E}=\mathbf{V}-\mathbf{B}$, one can write 
\begin{equation*}
\label{mt}
\begin{aligned}
&\mathbf{D}\mathbf{B}+\mathbf{A}\mathbf{E}+\mathbf{DE}+\mathbf{C}
\\&=(\mathbf{U}-\mathbf{A})\mathbf{B}+(\mathbf{V}-\mathbf{B})\mathbf{E}+(\mathbf{U}-\mathbf{A})(\mathbf{V}-\mathbf{B})+\mathbf{C}
\\&=\mathbf{UV}.
\end{aligned}
\end{equation*}
At the end of this step, client $i$ successfully obtain the secret share of the multiplication  for $\mathbf{U}$ and $\mathbf{V}$ as $[\mathbf{UV}]_i$, for $i\in[N]$. 
Note that for each multiplication between a pair of secret shares, we must use a new pair of analog triplets. We summarize the multiplication computation protocol in the analog domain in Algorithm \ref{Algo:analogmul}.
\end{itemize}

\begin{algorithm}[t]
\caption{Multiplication Computation in \texttt{A-MPC}}
\label{Algo:analogmul}
\begin{algorithmic}[1]
\REQUIRE{Number of clients $N$, number of colluding clients $T$, public parameters $\omega_j$'s for $j\in[N]$.}
\renewcommand{\algorithmicrequire}{\textbf{Input:}}
\renewcommand{\algorithmicensure}{\textbf{Output:}}
\REQUIRE{ $\mathbf{U}\in\mathbb{R}^{m_1\times n_1}$, $\mathbf{V}\in\mathbb{R}^{m_2\times n_2}$.}
\ENSURE  $\{[\mathbf{UV}]_i\}_{i\in[N]}$.\\
\vspace{2mm}
\textbf{Offline phase.}
\STATE The protocol randomly generate matrices $\mathbf{A}\in\mathbb{R}^{m_1\times n_1}$, $\mathbf{B}\in\mathbb{R}^{m_2\times n_2}$, where all entries in both $\mathbf{A}$ and $\mathbf{B}$ follow the distribution $\mathcal{TN}(0,\sigma^2;[-t,t])$, and $\mathbf{C}\in\mathbb{R}^{m_1\times n_2}$ given that $\mathbf{AB=C}$. Note that $n_1=m_2$. \\
\STATE Secret shares $\mathbf{A,B,C}$ to all clients so that client $i$ holds $[\mathbf{A}]_i$,$[\mathbf{B}]_i$, and $[\mathbf{C}]_i$, for $i\in[N]$.

\vspace{2mm}
\textbf{Online phase.} 
\STATE The protocol receives the inputs $\mathbf{U}\in\mathbb{R}^{m_1\times n_1}$, $\mathbf{V}\in\mathbb{R}^{m_2\times n_2}$.
\STATE Secretly shares the inputs $\mathbf{U},\mathbf{V}$ to all clients.\\
// Client $i$ now holds $[\mathbf{U}]_i,[\mathbf{V}]_i$, also, $[\mathbf{A}]_i,[\mathbf{B}]_i,[\mathbf{C}]_i$, for $i\in[N]$.
\STATE Client $i$ computes $[\mathbf{D}]_i=[\mathbf{U}]_i-[\mathbf{A}]_i$ and $[\mathbf{E}]_i=[\mathbf{V}]_i-[\mathbf{B}]_i$, for $i\in[N]$.
\STATE Client $i$ collects at least $T+1$ shares of $[\mathbf{D}]_j$'s and $[\mathbf{E}]_j$'s from clients $j\in[N]$ to reconstruct $\mathbf{D}$ and $\mathbf{E}$, for $i\in[N]$. \\
\STATE Client $i$ computes $[\mathbf{UV}]_i=\mathbf{D}[\mathbf{B}]_i+[\mathbf{A}]_i\mathbf{E}+\mathbf{DE}+[\mathbf{C}]_i$, for $i\in[N]$.
\end{algorithmic}
\end{algorithm}

\begin{remark}
Note that in the online phase for the analog multiplication triplets, the publicly revealed parameters $\mathbf{D}$ and $\mathbf{E}$ involve the input data matrices $\mathbf{U}$ and $\mathbf{V}$. Therefore, there is some privacy leakage in this particular step that needs to be carefully characterized. 
Since we randomly generate the entries of $\mathbf{A}$ and $\mathbf{B}$ according to $\mathcal{TN}(0,\sigma^2;[-t,t])$, 
the entries in $\mathbf{D}$ and $\mathbf{E}$ are, in a sense from the privacy-preserving perspective, in the same form as the secret shares specified in \eqref{eq:Sj}. Therefore, the same privacy analysis could be recycled for the analog multiplication triplets. The specific characterization of guarantees in terms of $(\epsilon,\delta)$-local differential privacy is provided later in the next section.
\end{remark}

For the sake of simplifying the analysis in the rest of this section, we suppose that each client holds scalar-valued data, i.e., the data at the $j$-th client is denoted by $x_j$, for $j\in [N]$. The secret share sent from the $j$-th client to the $i$-th client is denoted by $[x_j]_i=x_j+\sum_{k=1}^T\omega_i^k n_{j,k}$, for $i,j\in[N]$. 
The computation result at the $j$-th client is denoted by $y_j$, for $j\in[N]$. 
The share of the result held at the $i$-th client is denoted by $[y_j]_i$, for $i,j\in[N]$. 
The same analysis can be easily extended to data in the matrix form by applying it separately to all the entries of the data matrix.

Now, suppose that all the outputs $y_i$'s for $i\in [N]$ are linear combinations of the inputs $x_j$'s for $j\in [N]$,  i.e., 
\begin{equation}
y_i = a_{i,1}x_1+\cdots+a_{i,N}x_N=\sum_{l=1}^N a_{i,l}x_l,
\end{equation}
where $a_{i,k}$'s are constants. To recover the outcome, the $i$-th client has to gather $[y_i]_j$'s from clients $j\in [N]$, where $[y_i]_j = a_{i,1}[x_1]_j+\cdots+a_{i,N}[x_N]_j=\sum_{l=1}^N a_{i,l}[x_l]_j$, for $i,j\in[N]$. Note that in order to fully recover the secrets, one requires at least $T+1$ shares. Thus, let us denote the first $T+1$ shares received by the $i$-th client by $[y_i]_{([T+1])}=\left[ [y_i]_{(1)}\cdots[y_i]_{(T+1)}\right]^{\top}$, the corresponding evaluation points to these shares are denoted by $\omega_{(1)},\dots\omega_{(T+1)}$, respectively, for $i\in[N]$. Hence, one can write
\begin{equation}
\label{eq:system}
\begin{aligned}
[y_i]_{([T+1])}
=&
\begin{bmatrix}
\sum_{l=1}^N a_{i,l}[x_l]_{(1)} & \cdots & \sum_{l=1}^N a_{i,l}[x_l]_{(T+1)}
\end{bmatrix}^{\top}\\
=&
\begin{bmatrix}
1 & \omega_{(1)} & \cdots & \omega_{(1)}^{T}\\
\vdots & \vdots & \ddots & \vdots \\
1 & \omega_{(T+1)} & \cdots & \omega_{(T+1)}^{T}
\end{bmatrix}
\begin{bmatrix}
\sum_{l=1}^N a_{i,l}x_l \\
\sum_{l=1}^N a_{i,l}n_{l,1} \\
\vdots\\
\sum_{l=1}^N a_{i,l}n_{l,T}
\end{bmatrix}\\
= &
\,\mathbf{G}\begin{bmatrix}
z_0 & z_1 & \cdots & z_T
\end{bmatrix}^{\top}=\mathbf{Gz},
\end{aligned}
\end{equation}
where $\mathbf{z}=\begin{bmatrix}
z_0 & z_1 & \cdots & z_T
\end{bmatrix}^{\top}$, $z_0=y_i=\sum_{l=1}^N a_{i,l}x_l$, which is the secret $y_i$, $z_k=\sum_{l=1}^N a_{i,l}n_{l,k}$ for $k\in[T]$ and $i\in[N]$, and 
\begin{equation}
\mathbf{G}=\begin{bmatrix}
1 & \omega_{(1)} & \cdots & \omega_{(1)}^{T}\\
\vdots & \vdots & \ddots & \vdots \\
1 & \omega_{(T+1)} & \cdots & \omega_{(T+1)}^{T}
\end{bmatrix}.
\end{equation}
Note that the $i$-th client does not require the entire $\mathbf{z}$ but only $z_0$, which is the secret $y_i$, for $i\in [N]$. 

Let $\tilde{\mathbf{g}}$ denote the first row of $\mathbf{G}^{-1}$, the inverse of $\mathbf{G}$, which is well-defined since $\mathbf{G}$ is a Vandermonde matrix. Then, the $i$-th client only has to compute $\tilde{\mathbf{g}}[y_i]_{([T+1])}$ to recover the secret $y_i$ based on the first received $T+1$ shares, for $i\in[N]$. We summarize the computations discussed in this section in Algorithm \ref{Algo:A-MPC}.

\begin{algorithm}[t]
\caption{\texttt{A-MPC}}
\label{Algo:A-MPC}
\begin{algorithmic}[1]
\REQUIRE{Number of clients $N$, number of colluding clients $T$, public parameters $\omega_j$'s for $j\in[N]$.}
\renewcommand{\algorithmicrequire}{\textbf{Input:}}
\renewcommand{\algorithmicensure}{\textbf{Output:}}
\REQUIRE{Datasets $x_j$'s for clients $j\in[N]$.}
\ENSURE  {Computation results $y_i$'s of clients $i\in[N]$.}

\vspace{2mm}
\textbf{Secret sharing stage.} In this stage, each client receives shares of secrets of all the other clients, including themselves. Client $i$ holds $\{[x_1]_i,\dots,[x_N]_i\}$, for $i\in[N]$.

\vspace{2mm}
\textbf{Computation stage.} Repeat the following until all required computations in the given function have been processed.
\begin{itemize}
    \item \textbf{Addition computation.} We assume an analog addition for $u_1$ and $u_2$, where $u_1,u_2\in\mathbb{R}$. The $i$-th client holds secret shares for $u_1$ and $u_2$ as $[u_1]_i$ and $[u_2]_i$, for $i \in [N]$. Client $i$ computes $[u_1]_i+[u_2]_i$, for $i\in[N]$.
    \item \textbf{Multiply-by-a-constant computation.} We assume an analog multiply-by-a-constant for $u$, where $u\in\mathbb{R}$. The $i$-th client holds secret share for $u$ as $[u]_i$, for $i \in [N]$. We consider the multiplication of a constant $l$ and $u$. Client $i$ computes $l\cdot[u]_i$, for $i\in[N]$.
    \item \textbf{Multiplication computation.} We assume an analog multiplication for $u$ and $v$, where $u,v\in\mathbb{R}$. The $i$-th client holds secret shares for $u$ and $v$ as $[u]_i$ and $[v]_i$, for $i\in[N]$. Refer to Algorithm \ref{Algo:analogmul} with $u,v$ as the inputs, we can compute the secret share $[uv]_i$ for client $i$, for $i\in[N]$.
\end{itemize}

\textbf{Output reconstruction stage.} Client $i$, once collected at least $T+1$ shares, $[y_i]_{([T+1])}$'s, can start the reconstruction of the computation result by computing $y_i=\tilde{\mathbf{g}}[y_i]_{([T+1])}$, for $i\in[N]$.
\end{algorithmic}
\end{algorithm}

\subsubsection{Accuracy analysis}
By combining the proposed \emph{addition}, \emph{multiply-by-a-constant}, and \emph{multiplication triples} computations in the analog domain, the \texttt{A-MPC} protocol is thus capable of computing any polynomial function. It is worth noting that, \eqref{eq:addition}, \eqref{eq:multiplyconst}, and \eqref{eq:triples} imply that all computations require only linear computations of the secret shares in order to carry out the multiplication. In other words,  all the computations carried out over the secret shares are linear. The following theorem characterizes the perturbation in the computation outcome of the linear computations, establishing an upper limit on the worst-case computation error. It is assumed that the local computations performed by the clients do not impose any errors other than precision loss due to the finite representation of the results.
\begin{theorem}
Let $\Delta y_i$ denote the perturbation of $y_i$ in the protocol, for $i\in[N]$. Let $tT+r\geq 1$, $c=\sum_{l=1}^{N}|a_{i,l}|$, we have
\begin{equation}
\Delta y_i\leq c\sqrt{T+1}\cdot(r+tT)\frac{\kappa_G}{\lambda_{\mathrm{min}}}2^{-b_m},
\end{equation}
where $T$ denotes the maximum number of colluding clients, $t$ is the truncation parameter for the truncated Gaussian distribution, $\kappa_G$ is the condition number, $\lambda_{\mathrm{min}}$ is the minimum singular value of $\mathbf{G}$, $b_m$ is the number of precision bits, and $r$ is the bound on the absolute value of the secrets, for $i\in[N]$.
\end{theorem}
\begin{proof}
We have $[y_i]_{([T+1])}=\mathbf{G}\mathbf{z}$, for $i\in[N]$. Since the minimum singular value of $\mathbf{G}$ is given as $\lambda_\mathrm{min}$, then
\begin{equation}
\label{eq:sigular_ineq}
\lVert \mathbf{z} \rVert \leq \frac{\lVert [y_i]_{([T+1])} \rVert}{\lambda_\mathrm{min}},
\end{equation}
for $i\in[N]$.
Furthermore, one can write
\begin{equation}
\label{eq:y_upper}
\begin{aligned}
\lVert [y_i]_{([T+1])} \rVert &= (\sum_{k=1}^{T+1}\lvert [y_i]_{(k)} \rvert^2)^{\frac{1}{2}}\leq \sqrt{T+1}\max_{k\in [T+1]} \lvert [y_i]_{(k)} \rvert\\
&=\sqrt{T+1}\max_{k\in[T+1]}\lvert \sum_{l=1}^N a_{i,l}\cdot [x_l]_{(k)}\rvert \\
&\leq (\sum_{l=1}^{N}|a_{i,l}|)\cdot \sqrt{T+1}\max_{k\in[T+1]} \max_{l\in[N]}[x_l]_{(k)} \\
&= c\sqrt{T+1}\max_{k\in[T+1]} \max_{l\in[N]}(x_l+\sum_{j=1}^T\omega_{(k)}^j n_{l,j}) \\
&\leq c\sqrt{T+1}\cdot(r+tT),
\end{aligned}
\end{equation}
for $i\in[N]$.
Also, by noting that $[y_i]_{([T+1])}=\mathbf{G}\mathbf{z}$ as in \eqref{eq:system}, and that $y_i=\sum_{l=1}^N a_{i,l}x_l$ is an entry of $\mathbf{z}$, we have
\begin{equation}
\label{eq:perturb_y_z}
\Delta y_i \leq \lVert \Delta\mathbf{z} \rVert,
\end{equation}
for $i\in[N]$.
In \cite{analogss}, the relative perturbations of a system of linear equations $[y_i]_{([T+1])}=\mathbf{G}\mathbf{z}$ can be formulated as
\begin{equation}
\label{eq:relative_perturb_bound}
\frac{\lVert\Delta\mathbf{z}\rVert}{\lVert\mathbf{z}\rVert}\leq\kappa_{{G}}\frac{\lVert\Delta [y_i]_{([T+1])}\rVert}{\lVert [y_i]_{([T+1])} \rVert},
\end{equation}
for $i\in[N]$.
Combining \eqref{eq:sigular_ineq}, \eqref{eq:y_upper}, \eqref{eq:perturb_y_z}, together with \eqref{eq:relative_perturb_bound} results in
\begin{equation}
\label{eq:perturb_plug}
\frac{\Delta y_i\cdot\lambda_{\mathrm{min}}}{c\sqrt{T+1}\cdot(r+tT)}\leq\kappa_G\frac{\lVert\Delta [y_i]_{([T+1])}\rVert}{\lVert [y_i]_{([T+1])} \rVert},
\end{equation}
for $i\in[N]$.
Furthermore, since $b_m$ is the number of precision bits, we have
\begin{equation}
\label{eq:precision}
\frac{\lVert\Delta [y_i]_{([T+1])}\rVert}{\lVert [y_i]_{([T+1])} \rVert}\leq 2^{-b_m},
\end{equation}
for $i\in[N]$.
Combining \eqref{eq:perturb_plug} and \eqref{eq:precision} yields
\begin{equation}
\Delta y_i\leq c\sqrt{T+1}\cdot(r+tT)\frac{\kappa_G}{\lambda_{\mathrm{min}}}2^{-b_m},
\end{equation}
for $i\in[N]$.
\end{proof}





\section{ \texttt{A-MPC}: Analysis of Local Differential Privacy}
\label{sec:LDP}
In this section, we analyze the privacy of \texttt{A-MPC} against the worst case for colluding clients, i.e., when $T=N-1$. We focus on analyzing the privacy guarantees for each client in the secret sharing stage. First, the definition of $(\epsilon,\delta)$-local differential privacy for the protocol is provided. Then, we characterize the noise variance required in the protocol in order to satisfy the desired privacy level.

\begin{definition}[$(\epsilon,\delta)$-Local differential privacy for \texttt{A-MPC}]
Given a randomized mechanism $\mathcal{M}_{ji}:\mathbb{R}^{m\times n}\rightarrow{\mathbb{C}^{m\times n}}$, for $i,j\in[N]$. A protocol is $(\epsilon,\delta)$-locally differential private for the $j$-th client, for $j\in[N]$, if for the neighboring datasets $\mathbf{X}_j,\mathbf{X}_j^\prime\in\mathbb{R}^{m\times n}$, where $\mathrm{dist}(\mathbf{X}_j,\mathbf{X}_j^\prime)=1$, and $\mathcal{T}\subset \mathbb{C}^{m\times n}$,
\begin{equation}
\mathbb{P}(\mathcal{M}_{ji}(\mathbf{X}_j)\in\mathcal{T})\leq e^\epsilon\cdot\mathbb{P}(\mathcal{M}_{ji}(\mathbf{X}_j^\prime)\in\mathcal{T})+\delta,
\end{equation}
where $i,j\in[N]$.
\end{definition}
Note that the randomized mechanism $\mathcal{M}_{ji}$ is the mechanism for the $j$-th client sharing its secrets to the $i$-th client, i.e.,
\begin{equation}
\label{eq:random_mechanism}
\mathcal{M}_{ji}(\mathbf{X}_j)=[\mathbf{X}_j]_i=\mathbf{X}_j+\tilde{\mathbf{N}}_{ji}=\mathbf{X}_j+\sum_{k=1}^{T}\omega_i^k\mathbf{N}_{j,k},
\end{equation}
for $i,j\in[N]$. The definition of local sensitivity for the proposed \texttt{A-MPC} is given as follows.
\begin{definition}[Local sensitivity for \texttt{A-MPC}]
Let   $\mathbf{X}_j$ and $\mathbf{X}_j^\prime$ denote two neighboring datasets at the $j$-th client, where $\mathbf{X}_j,\mathbf{X}_j^\prime\in\mathbb{R}^{m\times n}$, for $j\in[N]$. The sensitivity is defined as 
\begin{equation}
\Delta\overset{\mathrm{def}}{=}\max_{\mathrm{dist}(\mathbf{X}_j,\mathbf{X}_j^\prime)=1}||\mathbf{X}_j-\mathbf{X}_j^\prime||_F,
\end{equation}
for $j\in[N]$.
\end{definition}
We formulate the absolute value of the privacy loss function
for scalar-valued computations in the following lemma.

\begin{lemma}
Consider a pair of neighboring datasets $d_j,d_j^\prime\in\mathbb{R}$ in the \texttt{A-MPC} protocol, where $d_j=d_j^\prime-\Delta$, for $j\in[N]$. Then the absolute value of the privacy loss is
\begin{equation}
|l_{\mathcal{M}_{ji},d,d^\prime}(y)|=|\frac{1}{2\sigma^2}(-2y\Delta+(\Delta)^2)|\cdot\mathbb{I}_{[-t+\Delta,t]}(y),
\end{equation}
for $i,j\in[N]$.
\end{lemma}
\begin{proof}
Note that one can write the PDFs of the perturbed output of the mechanisms as
\begin{equation}
\label{eq:pdfampcscalar1}
p_{\mathcal{M}_{ji}(d)}(y)=\frac{\frac{1}{\sqrt{2\pi\sigma^2}}\exp(-\frac{y^2}{2\sigma^2})}{2\Phi(\frac{t}{\sigma})-1}\cdot\mathbb{I}_{[-t,t]}(y),
\end{equation}
and,
\begin{equation}
\label{eq:pdfampcscalar2}
p_{\mathcal{M}_{ji}(d^\prime)}(y)=\frac{\frac{1}{\sqrt{2\pi\sigma^2}}\exp(-\frac{(y-\Delta)^2}{2\sigma^2})}{2\Phi(\frac{t}{\sigma})-1}\cdot\mathbb{I}_{[-t+\Delta,t+\Delta]}(y),
\end{equation}
for $i,j\in[N]$.
Then, by \eqref{eq:pdfampcscalar1} and \eqref{eq:pdfampcscalar2}, we obtain the absolute value of the privacy loss function as
\begin{equation}
\label{eq:plf-truncated_scalar}
\begin{aligned}
&|l_{\mathcal{M}_{ji},d,d^\prime}(y)|=|\mathrm{ln}(\frac{p_{\mathcal{M}_{ji}(d)}(y)}{p_{\mathcal{M}_{ji}(d^\prime)}(y)})|\\
=&|\frac{1}{2\sigma^2}(-2y\Delta+(\Delta)^2)|\cdot\mathbb{I}_{[-t+\Delta,t]}(y),
\end{aligned}
\end{equation}
where the indicator function $\mathbb{I}(\cdot)$ takes non-zero values in $[-t,t]\cap[-t+\Delta,t+\Delta]=[-t+\Delta,t]$.
\end{proof}

In the following theorem, we show that, under certain constraints,  \texttt{A-MPC} guarantees $(\epsilon,\delta)$-local differential privacy. 
\begin{theorem}
\label{thm:highdim}
The \texttt{A-MPC} protocol is $(\epsilon,\delta)$-locally differential private if $\mathbf{X}_j=\mathbf{X}_j^\prime+\mathbf{W}$ and 
\begin{equation}
\label{eq:LDPthm-AMPC}
0\leq 1-\frac{\Phi(\frac{\sigma\epsilon}{\Delta}+\frac{\Delta}{2\sigma})-\Phi(-\frac{\sigma\epsilon}{\Delta}+\frac{\Delta}{2\sigma})}{2\Phi(\frac{t}{\sigma})-1}\leq \delta,
\end{equation}
where $\mathbf{X}_j$ and $\mathbf{X}_j^\prime$ are a pair of neighboring datasets, $\mathbf{X}_j,\mathbf{X}_j^\prime,\mathbf{W}\in\mathbb{R}^{m\times n}$, $||\mathbf{W}||_F\leq\Delta$, and $\sigma\in (0,\sqrt{\frac{t\cdot\Delta}{\epsilon}-\frac{\Delta^2}{2\epsilon}})$.
\end{theorem}
\begin{proof}
The proof can be found in Appendix~\ref{app:highdim}.
\end{proof}

For the sake of simplifying the derivations in the rest of the paper, we introduce a variable $\alpha$ where $\sigma=\frac{\alpha \cdot \Delta}{\sqrt{2\epsilon}}$.
\begin{corollary}
Given $\epsilon, \delta, t, \Delta$, we can simplify \eqref{eq:LDPthm-AMPC} by introducing a variable $\alpha$ which satisfies $\sigma=\frac{\alpha \cdot \Delta}{\sqrt{2\epsilon}}$ as follows:
\begin{equation}
\label{eq:Bineq}
0\leq 1-\frac{\Phi(\sqrt{\frac{\epsilon}{2}}(\alpha+\frac{1}{\alpha}))-\Phi(\sqrt{\frac{\epsilon}{2}}(-\alpha+\frac{1}{\alpha}))}{2\Phi(\frac{t\sqrt{2\epsilon}}{\alpha\cdot \Delta})-1}\leq \delta,
\end{equation}
where $\alpha \in (0,\sqrt{\textstyle\frac{2t}{\Delta}-1}).$
\end{corollary}
\begin{proof}
In \eqref{eq:LDPthm-AMPC}, the term $-\frac{\sigma\epsilon}{\Delta}+\frac{\Delta}{2\sigma}$ in the numerator changes sign at $\frac{\sigma\epsilon}{\Delta}=\frac{\Delta}{2\sigma}$, where $\sigma=\frac{\Delta}{\sqrt{2\epsilon}}$. 
To ease the computation, we set $\sigma=\frac{\alpha \cdot \Delta}{\sqrt{2\epsilon}}$ to substitute the $\sigma$'s in \eqref{eq:LDPthm-AMPC} and \eqref{eq:sigmarange}, which yields \eqref{eq:Bineq} where $\alpha \in (0,\sqrt{\textstyle\frac{2t}{\Delta}-1})$.
\end{proof}
\noindent We define 
\begin{equation}
\label{eq:B}
B(\alpha)\deff 1-\frac{\Phi(\sqrt{\frac{\epsilon}{2}}(\alpha+\frac{1}{\alpha}))-\Phi(\sqrt{\frac{\epsilon}{2}}(-\alpha+\frac{1}{\alpha}))}{2\Phi(\frac{t\sqrt{2\epsilon}}{\alpha\cdot \Delta})-1},
\end{equation}
to simplify \eqref{eq:Bineq}, and note   that $0\leq B(\alpha)\leq \delta$.


We propose a truncated Gaussian mechanism of choosing the optimal noise variance $\sigma^2$ for the proposed \texttt{A-MPC} in terms of satisfying $(\epsilon,\delta)$-local differential privacy. 
\begin{theorem}[Analytical truncated Gaussian mechanism]
\label{thm:analytical}
The proposed \texttt{A-MPC} protocol satisfies $(\epsilon,\delta)$-local differential privacy for $\sigma=\frac{\alpha^* \cdot \Delta}{\sqrt{2\epsilon}}$, where $\alpha^*$ is obtained by:
\begin{equation}
\label{eq:alphaopt}
\begin{aligned}
\alpha^*=\,&\underset{\alpha}{\mathrm{argmax}}\, B(\alpha)\\
\mathrm{s.t.}\,\, &0\leq B(\alpha)\leq\delta,
\;
0<\alpha<\sqrt{\frac{2t}{\Delta}-1}.
\end{aligned}
\end{equation}
 
\end{theorem}

Theorem \ref{thm:analytical} characterizes an optimization problem aimed at determining the optimal noise standard deviation $\sigma$ to achieve  $(\epsilon, \delta)$-local differential privacy.

The following lemma  indicates that $B(\alpha)$ is a monotonically decreasing function of $\alpha$ for $\alpha\in (0,\sqrt{\frac{2t}{\Delta}-1})$.
\begin{lemma}
\label{lemma:decrease}
The function $B(\alpha)$, specified in \eqref{eq:B}, is monotonically decreasing in $\alpha$ for $\alpha\in (0,\sqrt{\frac{2t}{\Delta}-1})$.
\end{lemma}
\begin{proof}
The proof can be found in Appendix~\ref{app:decrease}.
\end{proof}


More specifically, we search for the value of $\alpha$ that minimizes the difference between $B(\alpha)$ and the given $\delta$, subject to the constraint $0 \leq B(\alpha) \leq \delta$, and then designate this value as $\alpha^*$. 
This simplifies the search of $\alpha^*$ so it can be done by a first-order iterative optimization algorithm, e.g., gradient ascent, for finding the global maximum. 
With the optimal value $\alpha^*$, we obtain the optimal noise standard deviation by setting $\sigma=\frac{\alpha^*\cdot\Delta}{\sqrt{2\epsilon}}$.
Note that $\sigma$ is obtained for the randomized mechanism $\mathcal{M}_{ji}(\mathbf{X}_j)$ in \eqref{eq:random_mechanism} where all entries in $\tilde{\mathbf{N}}_{ji}$'s are distributed according to $\mathcal{TN}(0,\sigma^2;[-t,t])$, for $i,j\in[N]$.
Furthermore, $\tilde{\mathbf{N}}_{ji}$'s are obtained by linear combinations of $\mathbf{N}_{j,k}$'s generated by clients $j\in[N]$, as $\tilde{\mathbf{N}}_{ji}=\sum_{k=1}^T \omega_i^k\mathbf{N}_{j,k}$ where all entries in $\mathbf{N}_{j,k}$'s are generated according to $\mathcal{N}(0,\sigma^2_s)$, for $i,j\in[N]$.
Therefore, one needs to determine $\sigma_s$ based on the given value of $\sigma$, as demonstrated in the following theorem.
\begin{theorem}
\label{thm:submatrices}
In order to guarantee $(\epsilon,\delta)$-local differential privacy for the \texttt{A-MPC} protocol, it is sufficient to set the noise variance $\sigma_s^2$ of each entry in $\mathbf{N}_{j,k}$'s in \eqref{eq:random_mechanism}, for $i,j\in[N]$, as 
\begin{equation}
\label{eq:sigma_submatrices}
\sigma^2_s=\frac{(\alpha^*)^2 \cdot \Delta^2}{2\epsilon T}.
\end{equation}
\end{theorem}
\begin{proof}
The proof can be found in Appendix~\ref{app:submatrices}.
\end{proof}


\begin{remark}
In \cite{copml}, a protocol, referred to as COPML, is proposed to tackle a similar scenario with perfect privacy guarantees, i.e., no information about the individual datasets as well as the intermediate model parameters is leaked in an information-theoretic sense. However, COPML has an upper bound on the maximum number of colluding clients $T$ that it can handle, which is bounded from above by $N/3$. 
In our proposed protocol, this threshold can be increased to as much as $T = N - 1$, demonstrating an improvement in the threshold by a factor of $3$. Note that the types of privacy guarantees are different in our protocol (differential privacy guarantee) versus COPML (perfect information-theoretic guarantee). Another major distinction between our protocol and COPML is that COPML can only be run over finite fields, which, in turn, necessitates quantizing real-valued datasets into finite field elements before running the protocol. On the contrary, our protocol, in theory, runs over real-valued datasets, and, in common practical settings, can be run over datasets represented by floating point numbers.
\end{remark}
\section{Experiments}
\label{sec:exp}
In this section, we present the experimental results of our protocol, employing specific training algorithms over a variety of datasets. The experiments focus on classification and regression tasks conducted in a fully-decentralized setting.
To ensure individual client privacy, we propose $(\epsilon,\delta)$-locally differential private logistic regression and linear regression algorithms based on the proposed \texttt{A-MPC}. 
A comparison is made between the training results obtained from our approach and those achieved through a centralized training without privacy guarantees.
The results demonstrate that by carefully tuning the noise parameters, the protocol achieves $(\epsilon,\delta)$-local differential privacy guarantees while closely following the accuracy of a centralized scheme with no privacy guarantee.



\subsection{Classification}
Building upon \texttt{A-MPC}, as summarized in Algorithm~\ref{Algo:A-MPC}, we propose an algorithm for training a logistic regression model in a $(\epsilon,\delta)$-locally differential private decentralized system for $N$ clients in Algorithms~\ref{algo:LR_AMPC} and \ref{algo:privatemul}. 
Note that logistic regression is an algorithm for binary classification, i.e., there are only two classes in the datasets. 

\begin{algorithm}[t]
\caption{$(\epsilon,\delta)$-Locally Differential Private Logistic Regression based on \texttt{A-MPC}}
\begin{algorithmic}[1]
\label{algo:LR_AMPC}
\REQUIRE{Privacy requirement ($\epsilon$,\,$\delta$), number of clients $N$, number of colluding clients $T$, public parameters $\omega_j$'s for $j\in[N]$, learning rate $\gamma$, iterations $J$, batch size $B$.}
\renewcommand{\algorithmicrequire}{\textbf{Input:}}
\renewcommand{\algorithmicensure}{\textbf{Output:}}
\REQUIRE{Datasets $\mathbf{X}_j$'s, label vectors $\mathbf{y}_j$'s, for clients $j\in[N]$.}
\ENSURE  Model weight vector $\mathbf{w}^{(J)}$.
\STATE Initialize model weight vector $\mathbf{w}^{(0)}$ randomly.
\STATE Secretly share $\mathbf{w}^{(0)}$ with clients $i\in[N]$.\\
// Client $i$ holds $[\mathbf{w}^{(0)}]_i$, for $i\in[N]$.

\FOR{$t=0,\dots,J-1$}
    
    \FOR {client $j=1,\dots,N$}
        \STATE Client $j$ randomly chooses $B$ datapoints, referred to $\mathbf{X}_j^B$'s and $\mathbf{y}_j^B$'s.
        \STATE Client $j$ secretly shares its partial individual dataset  $(\mathbf{X}_j^B,\mathbf{y}_j^B)$ with clients $i\in[N]$.
    \ENDFOR \\
    // Client $i$ holds $([\mathbf{X}_j^B]_i,[\mathbf{y}_j^B]_i)$ sent from client $j$, for $i,j\in[N]$.
    \STATE Client $i$ concatenates $\{[\mathbf{X}_j^B]_i\}_{j\in[N]}$  to $[\mathbf{X}^B]_i$, and $\{[\mathbf{y}_j^B]_i\}_{j\in[N]}$ to $[\mathbf{y}^B]_i$, for $i\in[N]$.

    \STATE All clients collaborate to compute $[\mathbf{X}^B\mathbf{w}^{(t)}]_i$ at client $i$, based on $[\mathbf{X}^B]_i$ and $[\mathbf{w}^{(t)}]_i$, for $i\in[N]$, according to $\{[\mathbf{X}^B\mathbf{w}^{(t)}]_i\}_{i\in[N]}=\mathrm{PrivateMul}(\{[\mathbf{X}^B]_i,[\mathbf{w}^{(t)}]_i\}_{i\in[N]})$.   
    \FOR{client $i=1,\dots,N$}
        \STATE Client $i$ computes $[\hat{g}(\mathbf{X}^B\mathbf{w}^{(t)})]_i=\frac{\mathbf{1}_B}{2}+\frac{1}{4}[\mathbf{X}^B\mathbf{w}^{(t)}]_i$.
        \STATE Client $i$ computes $[\mathbf{e}]_i=[\hat{g}(\mathbf{X}^B\mathbf{w}^{(t)})]_i-[\mathbf{y}^B]_i$.
    \ENDFOR

    \STATE All clients collaborate to compute $[(\mathbf{X}^B)^{\top}\mathbf{e}]_i$, at client $i$, based on $[\mathbf{X}^B]_i^\top$ and $[\mathbf{e}]_i$, for $i\in[N]$, according to $\{[(\mathbf{X}^B)^{\top}\mathbf{e}]_i\}_{i\in[N]}=\mathrm{PrivateMul}(\{[\mathbf{X}^B]_i^\top,[\mathbf{e}]_i\}_{i\in[N]})$. 
    \FOR{client $i=1,\dots,N$}
        \STATE $[\mathbf{w}^{(t+1)}]_i = [\mathbf{w}^{(t)}]_i - \frac{\gamma}{NB}[(\mathbf{X}^B)^{\top}\mathbf{e}]_i$.
    \ENDFOR
\ENDFOR
\FOR{client $j=1,\dots,N$}
    \STATE Collect at least $T+1$ secret shares from $\{[\mathbf{w}^{(J)}]_i\}_{i\in[N]}$ to reconstruct the trained model weight vector $\mathbf{w}^{(J)}$.
\ENDFOR
 \end{algorithmic}
 \end{algorithm}
\begin{algorithm}[t]
\label{algo:privatemul}
\caption{$\mathrm{PrivateMul}(\{[\mathbf{U}]_i,[\mathbf{V}]_i\}_{i\in[N]})$}
\begin{algorithmic}[1]
\label{algo:privatemul}
\renewcommand{\algorithmicrequire}{\textbf{Input:}}
\renewcommand{\algorithmicensure}{\textbf{Output:}}
\REQUIRE $\{[\mathbf{U}]_i,[\mathbf{V}]_i\}_{i\in[N]}$ 
\ENSURE  $\{[\mathbf{UV}]_i\}_{i\in[N]}$
\FOR{client $i=1,\dots,N$}
    \FOR{client $j=1,\dots,N$}
        \IF{$i=j$}
            \STATE Client $i$ computes $[\mathbf{U}]_i\times [\mathbf{V}]_i$ and secret shares to client $k$ for $k\in[N]$.\\
            // Client $k$ receives a secret share $[[\mathbf{U}]_i\times [\mathbf{V}]_i]_k$.
        \ELSE
            \STATE Use Algorithm~\ref{Algo:analogmul} to compute $[[\mathbf{U}]_i \times[\mathbf{V}]_j]_k$ for clients $k\in[N]$, where $[\mathbf{U}]_i$ and $ [\mathbf{V}]_j$ are the inputs for Algorithm~\ref{Algo:analogmul}. \\
            // Client $k$ receives a secret share $[[\mathbf{U}]_i \times[\mathbf{V}]_j]_k$, for $k\in[N]$.
        \ENDIF
    \ENDFOR
\ENDFOR\\
// Client $k$ holds secret shares $[[\mathbf{U}]_i \times[\mathbf{V}]_j]_k$ for $i,j,k\in[N]$.
\FOR{client $k=1,\dots,N$}
    \STATE Client $k$ computes  $\frac{1}{N^2}\sum_{i=1}^{N}\sum_{j=1}^{N}[[\mathbf{U}]_i\times[\mathbf{V}]_j]_k$ to obtain $[\mathbf{U}\mathbf{V}]_k$.
\ENDFOR
\end{algorithmic}
\end{algorithm}

Consider a fully decentralized system with $N$ clients. Given datasets $\mathbf{X}_j\in\mathbb{R}^{m\times n}$'s, denoting that client $j$ holds a dataset of $m$ samples and $n$ features, for $j \in [N]$. 
Then, the total number of samples held by all clients is $Nm$ and we denote the original dataset, before splitting, as $\mathbf{X}\in\mathbb{R}^{Nm\times n}$. 
Let $\mathbf{y}_j\in\{0,1\}^{m}$ denote the corresponding label vector for the dataset $\mathbf{X}_j$ held by client $j$, for $j\in[N]$.
We denote the original label vector as $\mathbf{y}\in\{0,1\}^{Nm}$, for $j \in [N]$. 
The goal is to compute the weight vectors of the model by iteratively minimizing the cross-entropy function using the following equation to update the weight vectors:
\begin{equation}
\label{eq:logisticR_weight_update}
\mathbf{w}^{(t+1)}=\mathbf{w}^{(t)}-\frac{\gamma}{N}\mathbf{X}^\top (g(\mathbf{X}\mathbf{w}^{(t)})-\mathbf{y}),
\end{equation}
where $\mathbf{w}^{(t)}\in\mathbb{R}^{n}$ represents the estimated weight vector in iteration $t$, $\gamma$ is the learning rate, and $g(x)=\frac{1}{1+\exp(-x)}$ is the sigmoid function that performs element-wise operations on the inputs.



Algorithms \ref{algo:LR_AMPC} and \ref{algo:privatemul} are tailored for training a logistic regression model using the proposed \texttt{A-MPC}, which offers $(\epsilon,\delta)$-local differential privacy guarantees.
In Algorithm \ref{algo:LR_AMPC}, a weight vector $\mathbf{w}^{(0)}$ is randomly initialized, and then secretly shared to all $N$ clients. Thus, client $i$ holds $[\mathbf{w}^{(0)}]_i$, for $i\in[N]$.
Then, in each iteration of the training, each client randomly chooses $B$ datapoints in a batch as $\mathbf{X}_j^B$'s and $\mathbf{y}_j^B$'s, for $j\in[N]$. 
All clients secretly share their chosen datasets with all other clients including themselves. 
At the end of the secret sharing stage, all clients concatenate the received secret shares, denoted by $[\mathbf{X}^B]_i$ and $[\mathbf{y}^B]_i$, for $i\in[N]$. 
To compute the multiplication of $[\mathbf{X}^B]_i$ and $[\mathbf{w}^{(t)}]_i$ as $[\mathbf{X}^B \mathbf{w}^{(t)}]_i$, for $i\in[N]$, while keeping all included datasets and model parameters private, we propose Algorithm \ref{algo:privatemul}. The detailed descriptions of Algorithm \ref{algo:privatemul} are moved to Appendix~\ref{app:privatemul}. 
Steps $10$ through $17$ involve updating the model parameters, following the equation \eqref{eq:logisticR_weight_update}.
Note that the sigmoid function in~\eqref{eq:logisticR_weight_update}, $g(x)$, is substituted by its degree-$1$ polynomial approximation as $\hat{g}(x)= \frac{1}{2}+\frac{x}{4}$ during the training. This procedure is repeated until the given number of iterations $J$ is reached. The last update is the final result for the logistic regression model as a weight vector $\mathbf{w}^{(J)}$.

We train the logistic regression model over datasets such as MNIST on the digits $2$ and $6$ \cite{mnist}, Titanic \cite{titanic} and Cleveland Heart Disease \cite{Dua:2019} referred to as Heart. 
In each dataset, we split the samples equally to all clients. 
Consider datasets $\mathbf{X}_j\in\mathbb{R}^{m\times n}$, denoting that each client holds a dataset of $m$ samples and $n$ features, for $j \in [N]$. 
Then, the total number of samples held by all clients is $Nm$ and we denote the original dataset, before splitting, as $\mathbf{X}\in\mathbb{R}^{Nm\times n}$. We have: (1) $(Nm,n)=(10000,784)$ for MNIST; (2) $(Nm,n)=(200,13)$ for Heart; and (3) $(Nm,n)=(1000,26)$ for Titanic.  

Figure~\ref{fig:acc_two} shows the accuracy for the three datasets based on the proposed differentially private logistic regression. Figure~\ref{fig:acc_a}, Fig.~\ref{fig:acc_b}, and Fig.~\ref{fig:acc_f} show the numerical results of the experiments for the datasets with $N=2$ and $T=1$. We consider three different settings for the standard deviation of the added noises based on our proposed protocol for each individual dataset.
One may observe that for the smallest value of $\sigma$, i.e., the results in black curves, both datasets can follow the accuracy of the centralized training protocol closely with a negligible loss. 
In contrast, as $\sigma$ grows larger, the accuracy for both datasets suffers performance degradation. In the largest $\sigma$ we pick for both datasets, the accuracy approaches random guessing, as expected. Therefore, there exists a trade-off between accuracy and privacy, which one may make the decision based on the required privacy guarantee. 
The specific values of the local differential privacy parameters, namely $\epsilon$ and $\delta$, are indicated in the captions corresponding to the plots presented in Fig.\,\ref{fig:acc_two}.

For the experiments with a larger number of clients, i.e., $N$, also follow the results. Figure~\ref{fig:acc_c}, Fig.~\ref{fig:acc_d}, and Fig.~\ref{fig:acc_e} show the experiments of the Heart dataset with $N=4$, while having different numbers of colluding clients $T=1$, $T=2$, and $T=3$, respectively. Figure~\ref{fig:acc_g} and Fig.~\ref{fig:acc_h} show the experiments of the Titanic dataset with $N=10$, while having different numbers of colluding clients $T=1$ and $T=9$, respectively. The experiments show that with a larger $N$ or $T$, the datasets can still align the accuracy with the centralized setting. Therefore, by setting the noise parameters carefully, the protocol provides $(\epsilon,\delta)$-local differential privacy guarantees at the cost of negligible accuracy loss.  
\begin{figure*}[t]
\captionsetup{justification=centering}
\centering
\subfloat[\footnotesize MNIST: $(N,T)=(2,1)$, $(\epsilon,\delta)\approx(10^{-2},10^{-5})$]{%
\includegraphics[trim=25 2 30 10,clip,width=0.23\textwidth]{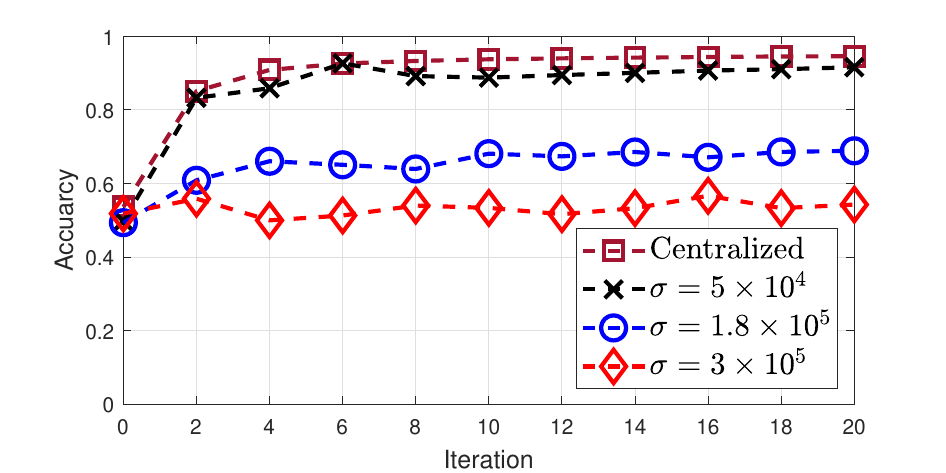}\label{fig:acc_a}}\hfill
\subfloat[{\footnotesize Heart: $(N,T)=(2,1)$, $(\epsilon,\delta)\approx(10^{-3},10^{-8})$}]{
\includegraphics[trim=25 2 30 10,clip,width=0.23\textwidth]{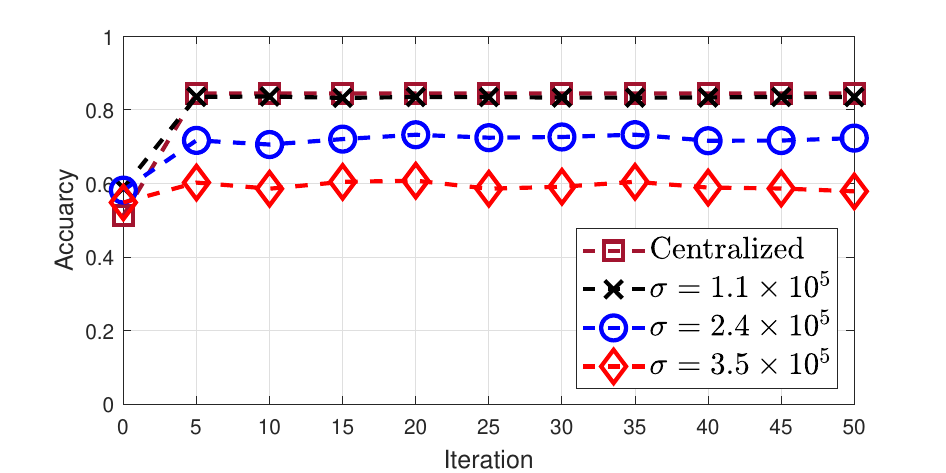}\label{fig:acc_b}}\hfill
\subfloat[\footnotesize Heart: $(N,T)=(4,1)$, $(\epsilon,\delta)\approx(10^{-3},10^{-8})$]{%
\includegraphics[trim=25 2 30 10,clip,width=0.23\textwidth]{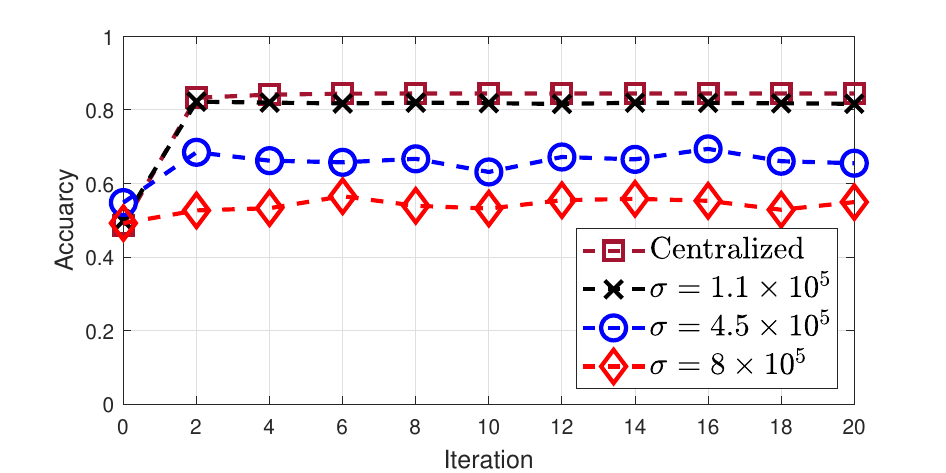}\label{fig:acc_c}}\hfill
\subfloat[\footnotesize Heart: $(N,T)=(4,2)$, $(\epsilon,\delta)\approx(10^{-3},10^{-8})$]{%
\includegraphics[trim=25 2 30 10,clip,width=0.23\textwidth]{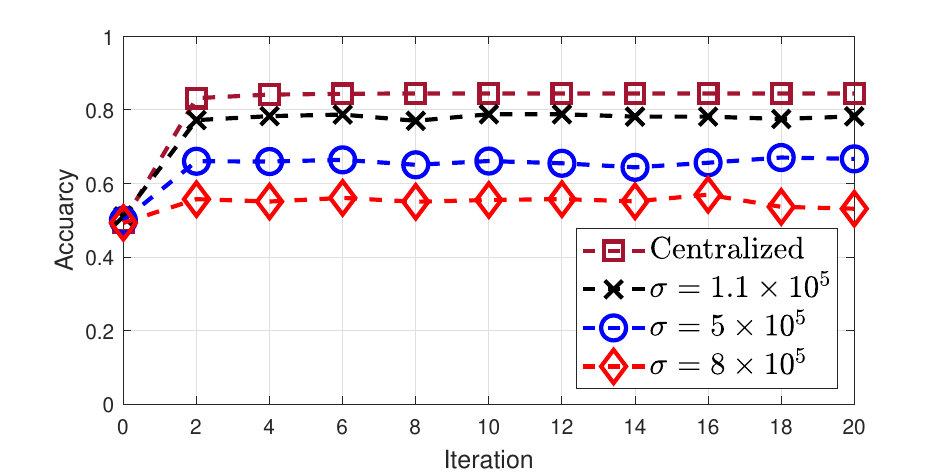}\label{fig:acc_d}}
\\
\subfloat[\footnotesize Heart: $(N,T)=(4,3)$, $(\epsilon,\delta)\approx(10^{-3},10^{-8})$]{%
\includegraphics[trim=25 2 30 10,clip,width=0.23\textwidth]{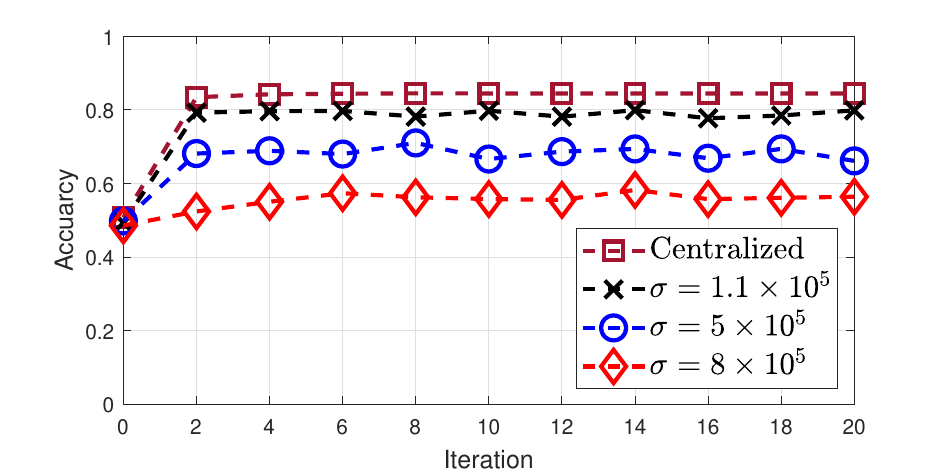}\label{fig:acc_e}}\hfill
\subfloat[\footnotesize Titanic: $(N,T)=(2,1)$, $(\epsilon,\delta)\approx(5\times10^{-3},10^{-8})$]{%
\includegraphics[trim=25 2 30 10,clip,width=0.23\textwidth]{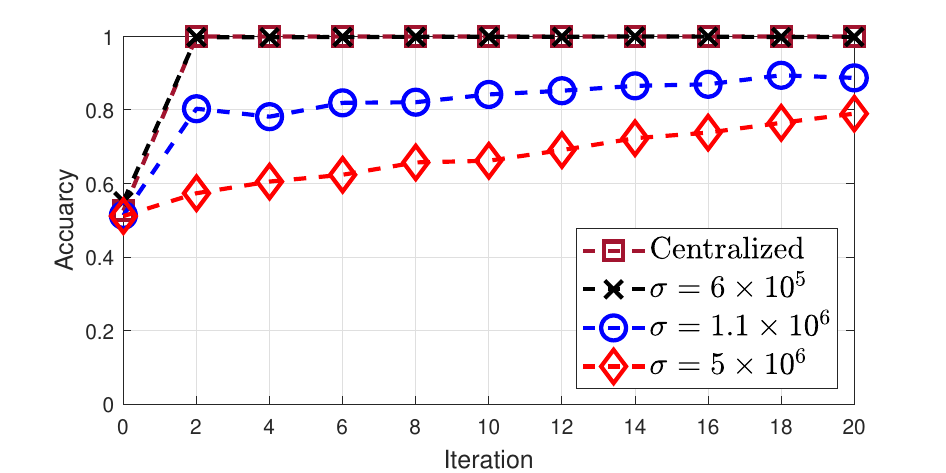}\label{fig:acc_f}}\hfill
\subfloat[\footnotesize Titanic: $(N,T)=(10,1)$, $(\epsilon,\delta)\approx(5\times10^{-3},10^{-8})$]{%
\includegraphics[trim=25 2 30 10,clip,width=0.23\textwidth]{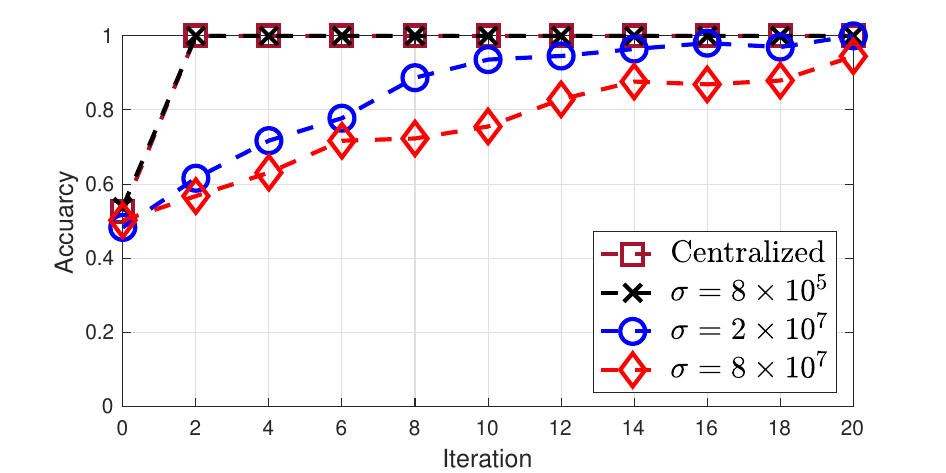}\label{fig:acc_g}}\hfill
\subfloat[\footnotesize Titanic: $(N,T)=(10,9)$, $(\epsilon,\delta)\approx(5\times10^{-3},10^{-8})$]{\includegraphics[trim=25 2 30 10,clip,width=0.23\textwidth]{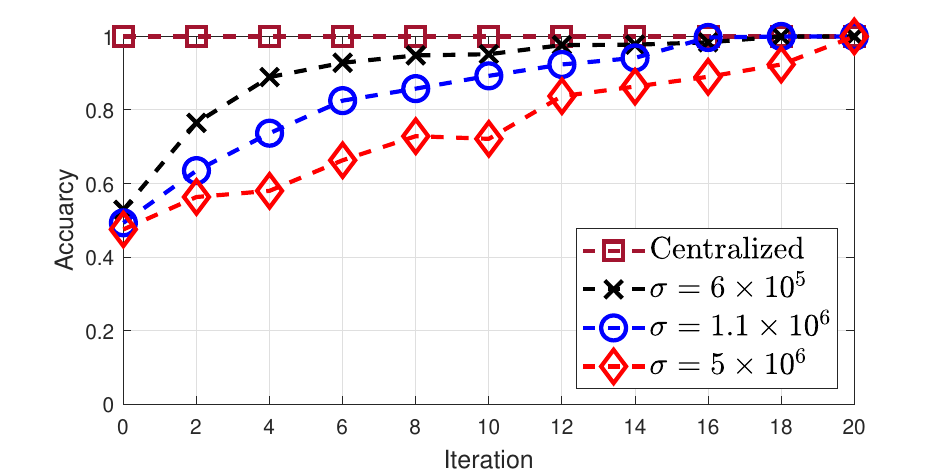}\label{fig:acc_h}}
\\
\caption{Accuracy for the datasets based on the proposed differentially private logistic regression training in the analog domain}\label{fig:acc_two}
\end{figure*}

\subsection{Regression}

Building upon \texttt{A-MPC}, as summarized in Algorithm~\ref{Algo:A-MPC}, we propose an algorithm for training a linear regression model in a $(\epsilon,\delta)$-locally differential private decentralized system for $N$ clients in Algorithms~\ref{algo:privatemul} and \ref{algo:LinearR_AMPC}.

Consider a fully decentralized system with $N$ clients. Given datasets $\mathbf{X}_j\in\mathbb{R}^{m\times n}$'s, denoting that client $j$ holds a dataset of $m$ samples and $n$ features, for $j \in [N]$. 
Then, the total number of samples held by all clients is $Nm$ and we denote the original dataset, before splitting, as $\mathbf{X}\in\mathbb{R}^{Nm\times n}$. 
Let $\mathbf{y}_j\in\mathbb{R}^{m}$ denote the corresponding label vector for the dataset $\mathbf{X}_j$ held by client $j$, for $j\in[N]$.
\begin{algorithm}[t]
\caption{$(\epsilon,\delta)$-Locally Differential Private Linear Regression based on \texttt{A-MPC}}
\begin{algorithmic}[1]
\label{algo:LinearR_AMPC}
\REQUIRE{Privacy requirement ($\epsilon$,$\delta$), number of clients $N$, number of colluding clients $T$, public parameters $\omega_j$'s for $j\in[N]$, learning rate $\gamma$, iterations $J$, batch size $B$.}
\renewcommand{\algorithmicrequire}{\textbf{Input:}}
\renewcommand{\algorithmicensure}{\textbf{Output:}}
\REQUIRE{Datasets $\mathbf{X}_j$'s, label vectors $\mathbf{y}_j$'s, for clients $j\in[N]$.}
\ENSURE  Model weight vector $\mathbf{w}^{(J)}$.
\STATE Initialize model weight vector $\mathbf{w}^{(0)}$ randomly.
\STATE Secret share $\mathbf{w}^{(0)}$ with client $i$, for $i\in[N]$. \\// Client $i$ holds $[\mathbf{w}^{(0)}]_i$, for $i\in[N]$.

\FOR {client $j=1,\dots,N$}
        \STATE Client $j$ concatenates an all-one vector $\mathbf{1}_{m}$ with its dataset $ \mathbf{X}_j$ as $\mathbf{X}_j^{\mathrm{cat}}=[\mathbf{1}_{m},\mathbf{X}_j]$.
    \ENDFOR \\

\FOR{$t=0,\dots,J-1$}
    
    \FOR {client $j=1,\dots,N$}
        \STATE Client $j$  randomly chooses $B$ datapoints from $\mathbf{X}_j^{\mathrm{cat}}$, referred to $\mathbf{X}_j^B$'s and $\mathbf{y}_j^B$'s.
        \STATE Client $j$ secretly shares its partial individual dataset  $(\mathbf{X}_j^B,\mathbf{y}_j^B)$ with clients $i\in[N]$.
    \ENDFOR \quad // Client $i$ holds $([\mathbf{X}_j^B]_i,[\mathbf{y}_j^B]_i)$ sent from client $j$, for $i,j\in[N]$.
    \STATE Client $i$ concatenates $\{[\mathbf{X}_j^B]_i\}_{j\in[N]}$  to $[\mathbf{X}^B]_i$, and $\{[\mathbf{y}_j^B]_i\}_{j\in[N]}$ to $[\mathbf{y}^B]_i$, for $i\in[N]$.

    \STATE All clients collaborate to compute $[\mathbf{X}^B\mathbf{w}^{(t)}]_i$ at client $i$, based on $[\mathbf{X}^B]_i$ and $[\mathbf{w}^{(t)}]_i$, for $i\in[N]$, according to $\{[\mathbf{X}^B\mathbf{w}^{(t)}]_i\}_{i\in[N]}=\mathrm{PrivateMul}(\{[\mathbf{X}^B]_i,[\mathbf{w}^{(t)}]_i\}_{i\in[N]})$.\\
    // Client $i$ holds $[\mathbf{X}^B\mathbf{w}^{(t)}]_i$ and $[\mathbf{y}^B]_i$, for $i\in[N]$.   
    \FOR{client $i=1,\dots,N$}
        \STATE Client $i$ computes $[\mathbf{e}]_i=[\mathbf{X}^B\mathbf{w}^{(t)}]_i-[\mathbf{y}^B]_i$.
    \ENDFOR

    \STATE All clients collaborate to compute $[(\mathbf{X}^B)^{\top}\mathbf{e}]_i$ at client $i$, based on $[(\mathbf{X}^B)^{\top}]_i$ and $[\mathbf{e}]_i$, for $i\in[N]$, according to $\{[(\mathbf{X}^B)^{\top}\mathbf{e}]_i\}_{i\in[N]}=\mathrm{PrivateMul}(\{[\mathbf{X}^B]_i^\top,[\mathbf{e}]_i\}_{i\in[N]})$.\\
    // Client $i$ holds $[\mathbf{w}^{(t)}]_i$ and $[(\mathbf{X}^B)^{\top}\mathbf{e}]_i$, for $i\in[N]$.
    
    \FOR{client $i=1,\dots,N$}
        \STATE $[\mathbf{w}^{(t+1)}]_i = [\mathbf{w}^{(t)}]_i - \frac{\gamma}{NB}[(\mathbf{X}^B)^{\top}\mathbf{e}]_i$.
    \ENDFOR
\ENDFOR
\FOR{client $j=1,\dots,N$}
    \STATE Collect at least $T+1$ secret shares from $\{[\mathbf{w}^{(J)}]_i\}_{i\in[N]}$ to reconstruct $\mathbf{w}^{(J)}$.
\ENDFOR
 \end{algorithmic}
 \end{algorithm}
We denote the original label vector as $\mathbf{y}\in\mathbb{R}^{Nm}$, for $j \in [N]$. 
The goal is to compute the weight vectors of the model by iteratively minimizing the function using the following equation to update the weight vectors:
\begin{equation}
\label{eq:linearR_weight_update}
\mathbf{w}^{(t+1)}=\mathbf{w}^{(t)}-\frac{\gamma}{N}\mathbf{X}^\top (\mathbf{X}\mathbf{w}^{(t)}-\mathbf{y}),
\end{equation}
where $\mathbf{w}^{(t)}\in\mathbb{R}^{n}$ is the estimated weight vector in iteration $t$, $\gamma$ is the learning rate.

We describe the steps in the experiment for training a model for linear regression. 
Algorithms \ref{algo:privatemul} and \ref{algo:LinearR_AMPC} are proposed for training a linear regression model based on the proposed \texttt{A-MPC} that provides $(\epsilon,\delta)$-local differential privacy guarantees. 
In Algorithm \ref{algo:LinearR_AMPC}, a weight vector $\mathbf{w}^{(0)}$ is randomly initialized, and then secretly shared to $N$ clients. 
Thus, client $i$ holds $[\mathbf{w}^{(0)}]_i$, for $i\in[N]$.
Then, in each iteration of the training,  client $j$ concatenates an all-one vector $\mathbf{1}_{m}$ with its dataset $ \mathbf{X}_j$ as $\mathbf{X}_j^{\mathrm{cat}}=[\mathbf{1}_{m},\mathbf{X}_j]$, for $j\in[N]$. 
Then, randomly chooses $B$ datapoints from $\mathbf{X}_j^{\mathrm{cat}}$ in a batch as $\mathbf{X}_j^B$'s and $\mathbf{y}_j^B$'s, for $j\in[N]$. 
All clients secretly share their chosen datasets with all other clients including themselves. 
At the end of the secret sharing stage, all clients concatenate the received secret shares, denoted by $[\mathbf{X}^B]_i$ and $[\mathbf{y}^B]_i$, for $i\in[N]$. 
To compute the multiplication of $[\mathbf{X}^B]_i$ and $[\mathbf{w}^{(t)}]_i$ as $[\mathbf{X}^B \mathbf{w}^{(t)}]_i$, for $i\in[N]$, while keeping all included datasets and model parameters secret, we use Algorithm \ref{algo:privatemul}. 
Steps $10$ through $17$ involve updating the model parameters, following the equation \eqref{eq:logisticR_weight_update}.
The last update is the weight vector for the linear regression model denoted by $\mathbf{w}^{(J)}$.

We train the linear regression model over datasets such as Combined Cycle Power Plant \cite{Dua:2019} referred to as CCPP, Red Wine Quality \cite{Dua:2019} referred to as Wine, Real Estate Price Prediction \cite{estate} referred to as Estate, and Tesla stock data from 2010 to 2020 \cite{tesla} referred to as Tesla. 
We distribute the samples evenly among all clients within each dataset. We have: (1) $(Nm,n)=(9568,4)$ for CCPP; (2) $(Nm,n)=(2000,12)$ for Wine; (3) $(Nm,n)=(414,6)$ for Estate; and (4) $(Nm,n)=(2417,5)$ for Tesla.
The performance is evaluated by relative error, which is defined by $e_{\mathrm{rel}}=\frac{||\mathbf{y}-\hat{\mathbf{y}}||}{||\mathbf{y}||}$, where $\mathbf{y}$ is the vector of true labels and $\hat{\mathbf{y}}$ is the vector of predicted labels. 

Figure~\ref{fig:re} shows the accuracy for the datasets based on the proposed $(\epsilon,\delta)$-locally differential private linear regression. Figure~\ref{fig:re_a}, Fig.~\ref{fig:re_b}, Fig.~\ref{fig:re_c}, and Fig.~\ref{fig:re_d} show the numerical results of the experiments for the datasets with $N=2$ and $T=1$. Similar to the logistic regression experiments, for the smallest value of $\sigma$, all datasets can follow the accuracy of the centralized training protocol closely with a negligible loss. 

\begin{figure*}[t]
\captionsetup{justification=centering}
\centering
\subfloat[\footnotesize CCPP: $(N,T)=(2,1)$, $(\epsilon,\delta)\approx(10^{-4},10^{-8})$]{%
\includegraphics[trim=25 2 30 10,clip,width=0.23\textwidth]{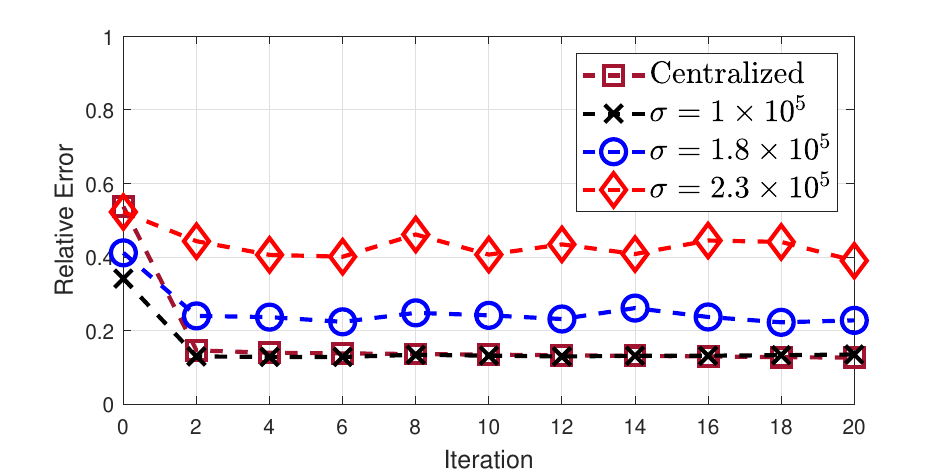}\label{fig:re_a}}\hfill
\subfloat[{\footnotesize Wine: $(N,T)=(2,1)$, $(\epsilon,\delta)\approx(2\times10^{-4},10^{-8})$}]{
\includegraphics[trim=25 2 30 10,clip,width=0.23\textwidth]{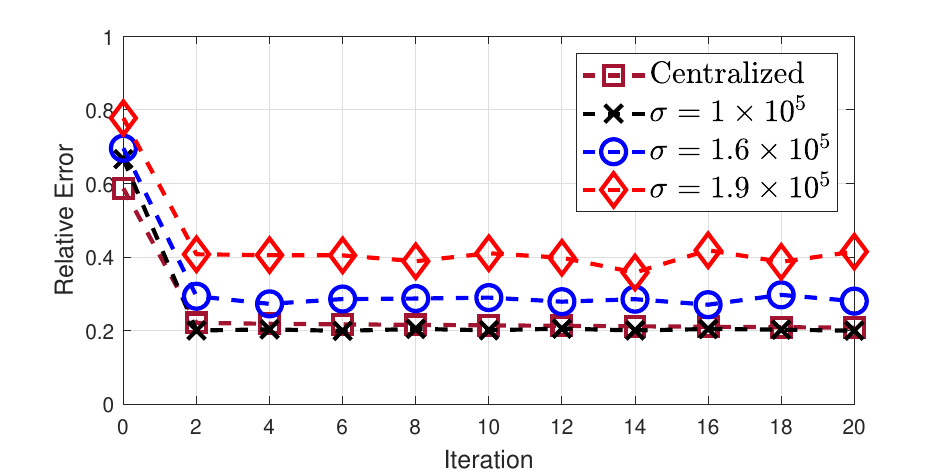}\label{fig:re_b}}\hfill
\subfloat[\footnotesize Estate: $(N,T)=(2,1)$, $(\epsilon,\delta)\approx(10^{-4},10^{-8})$]{%
\includegraphics[trim=25 2 30 10,clip,width=0.23\textwidth]{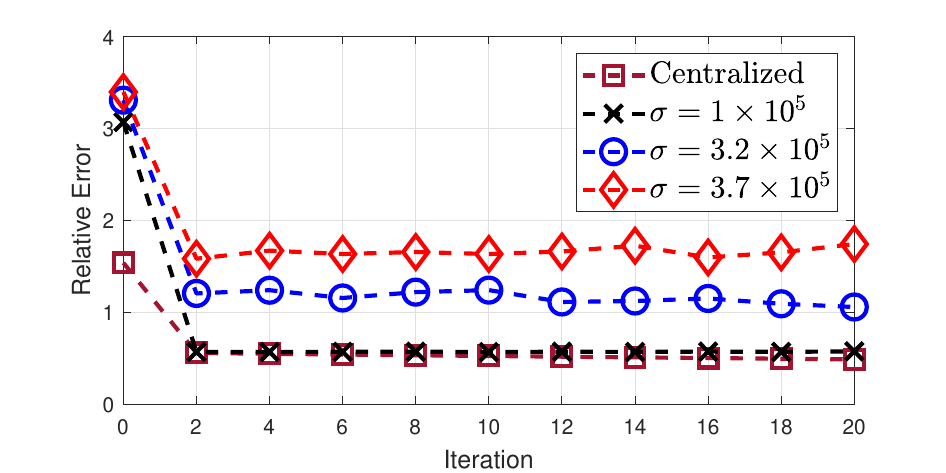}\label{fig:re_c}}\hfill
\subfloat[\footnotesize Tesla: $(N,T)=(2,1)$, $(\epsilon,\delta)\approx(3\times10^{-4},10^{-8})$]{%
\includegraphics[trim=25 2 30 10,clip,width=0.23\textwidth]{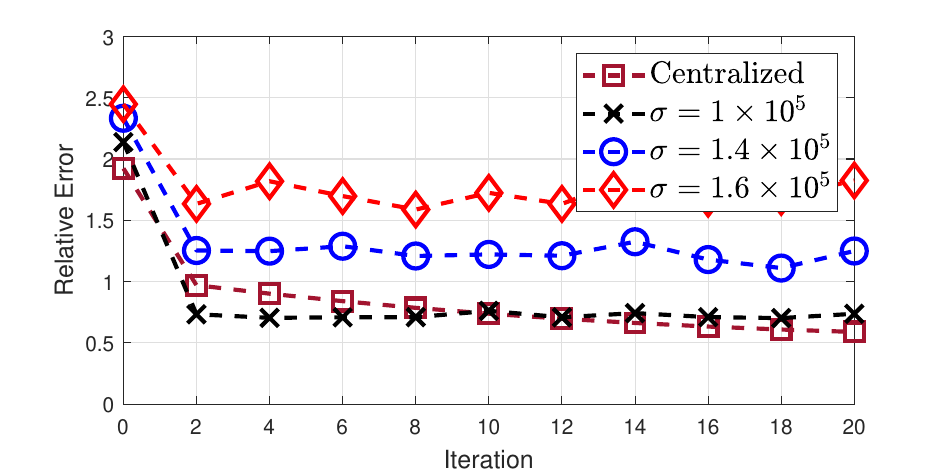}\label{fig:re_d}}
\\
\subfloat[\footnotesize CCPP: $(N,T)=(4,1)$, $(\epsilon,\delta)\approx(10^{-4},10^{-8})$]{%
\includegraphics[trim=25 2 30 10,clip,width=0.23\textwidth]{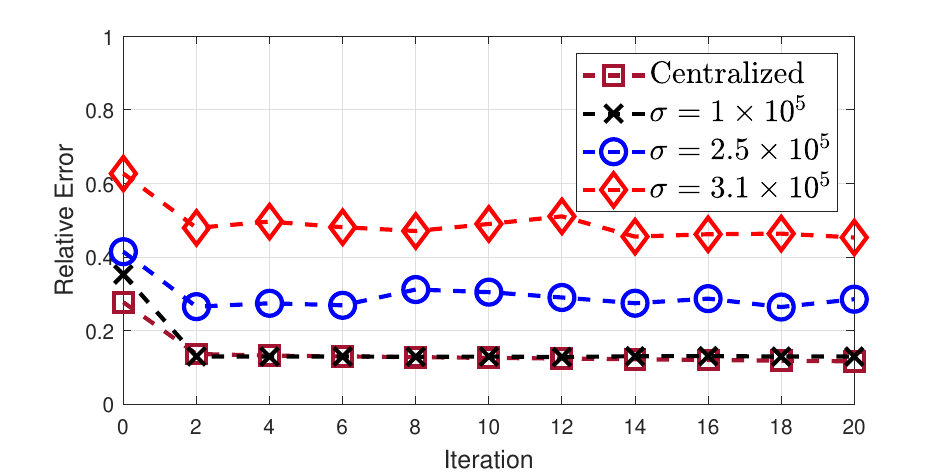}\label{fig:re_e}}\hfill
\subfloat[\footnotesize CCPP: $(N,T)=(4,3)$, $(\epsilon,\delta)\approx(10^{-4},10^{-8})$]{%
\includegraphics[trim=25 2 30 10,clip,width=0.23\textwidth]{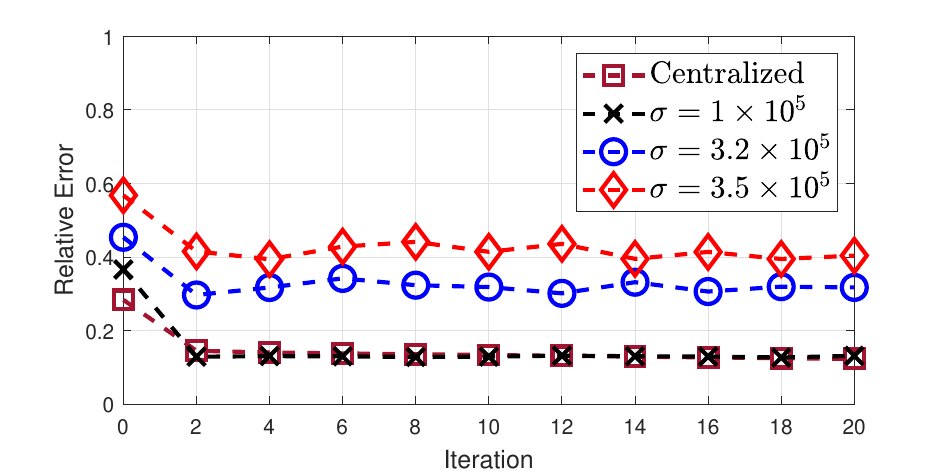}\label{fig:re_f}}\hfill
\subfloat[\footnotesize Wine: $(N,T)=(10,1)$, $(\epsilon,\delta)\approx(2\times10^{-4},10^{-8})$]{%
\includegraphics[trim=25 2 30 10,clip,width=0.23\textwidth]{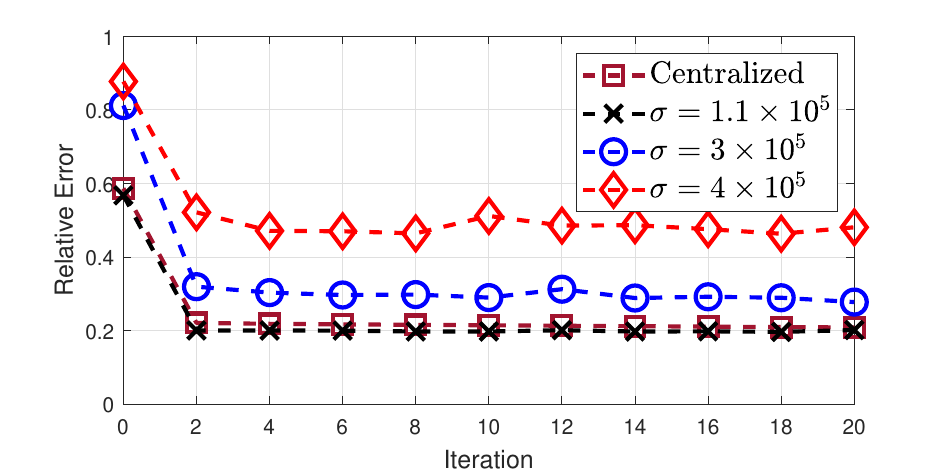}\label{fig:re_g}}\hfill
\subfloat[\footnotesize Wine: $(N,T)=(10,9)$, $(\epsilon,\delta)\approx(2\times10^{-4},10^{-8})$]{\includegraphics[trim=25 2 30 10,clip,width=0.23\textwidth]{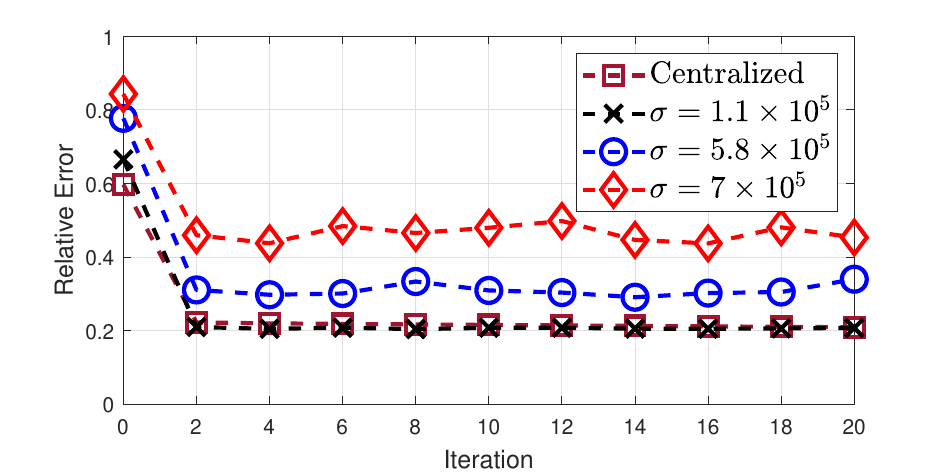}\label{fig:re_h}}\\
\caption{Relative error for the datasets based on the proposed differentially private linear regression training in the analog domain}
\label{fig:re}
\end{figure*}

Experiments on the CCPP dataset with $N=4$ are depicted in Figure~\ref{fig:re_e} and Fig.~\ref{fig:re_f}, showcasing scenarios with varying numbers of colluding clients ($T=1$ and $T=3$, respectively).
Experiments on the Wine dataset with $N=10$ are depicted in Figure~\ref{fig:re_g} and Fig.~\ref{fig:re_h}, illustrating scenarios with varying numbers of colluding clients ($T=1$ and $T=9$, respectively).
The results indicate that by increasing either $N$ or $T$, the datasets are able to closely match the accuracy of the centralized setting, incurring only a negligible loss in accuracy.
For small values of $\sigma$, the performance over both datasets closely tracks the accuracy of the centralized training protocol.
 In contrast, with the increase of $\sigma$, the accuracy of all datasets undergoes a decline in performance.
Therefore, a trade-off arises between accuracy and privacy, leading to decisions one has to make based on the desired privacy guarantee.
The specific values of the local differential privacy parameters, namely $\epsilon$ and $\delta$, are indicated in the captions corresponding to the plots presented in Fig.\,\ref{fig:re}.
\section{Conclusion}
\label{sec:CCS}
In this paper, we have proposed a fully distributed protocol \texttt{A-MPC} which allows computations in the analog domain among multiple clients without a trusted center while keeping the datasets of each client private. 
The proposed protocol includes \emph{addition}, \emph{multiply-by-a-constant}, and \emph{multiplication}, where the multiplication computation requires two phases of computations, i.e., \emph{offline phase} and \emph{online phase}. We have derived an upper bound on the perturbation considering a finite representation of the results.
As perfect privacy can no longer be guaranteed with MPC in the analog domain, we analyze the proposed protocol through the lens of local differential privacy.
With a formulated optimization problem based on given constraints, we may ensure the proposed \texttt{A-MPC} is $(\epsilon,\delta)$-locally differential private.  
Furthermore, we have proposed collaborative machine learning algorithms for training logistic regression and linear regression models based on our proposed \texttt{A-MPC} protocol. 
The experimental results demonstrate that the accuracy of the collaborative machine learning algorithms proposed in this paper, trained on real-world datasets, closely follows that of the centralized training conducted by a single trusted entity, all while maintaining the privacy of local datasets.


\appendix
\subsection{Proof of Theorem~\ref{thm:highdim}}
\label{app:highdim}
\begin{proof}
To simplify the notation, we replace $\mathbf{X}_j$ by $\mathbf{S}$ and replace $\mathbf{X}_j^\prime$ by $\mathbf{S}^\prime$, i.e., $\mathbf{S}=\mathbf{X}_j$ and $\mathbf{S}^\prime=\mathbf{X}_j^\prime$. We reformulate the matrices $\mathbf{S}$, $\mathbf{S}^\prime$, and $\mathbf{W}$ by vectorization, to $\mathrm{vec}(\mathbf{S})$, $\mathrm{vec}(\mathbf{S}^\prime)$, and $\mathrm{vec}(\mathbf{W})$, respectively, where $\mathrm{vec}(\mathbf{S}),\mathrm{vec}(\mathbf{S}^\prime),\mathrm{vec}(\mathbf{W})\in\mathbb{R}^{mn}$. The relation $\mathbf{S}=\mathbf{S}^\prime+\mathbf{W}$ becomes $\mathrm{vec}(\mathbf{S})=\mathrm{vec}(\mathbf{S}^\prime)+\mathrm{vec}(\mathbf{W})$. Also, we have $\lVert\mathbf{W}\rVert_F=\lVert\mathrm{vec}(\mathbf{W})\rVert\leq\Delta$. Let $\mathbf{s}=\mathrm{vec}(\mathbf{S})$.

Given an ordered basis $\mathcal{B}=\{\mathbf{b}_1,\mathbf{b}_2,...,\mathbf{b}_{mn}\}$, for every $\mathbf{s}\in\mathbb{R}^{mn}$, there is a unique linear combination of the basis vectors that is equal to $\mathbf{s}$. The representation of $\mathbf{s}$ in  $\mathcal{B}$ basis is the sequence of coordinates $[\mathbf{s}]_{\mathcal{B}}=[\eta_1,\eta_2,...\eta_{mn}]$, where $\eta_l\sim\mathcal{TN}(0,\sigma^2;[-t,t])$, for $l\in[mn]$. 
Define $\mathbf{s}^{[l]}\deff \eta_l\mathbf{b}_l$, then we have $\mathbf{s}=\sum_{l=1}^{mn}\mathbf{s}^{[l]}$.
The PDFs of the neighboring datasets are characterized as 
\begin{equation}
\label{eq:pdfampc1}
\begin{aligned}
p_{\mathcal{M}_{ji}(\mathbf{S})}(\mathbf{s})&=\frac{\frac{1}{(\sqrt{2\pi})^{\frac{mn}{2}}|\mathbf{\Sigma}|^{\frac{1}{2}}}\exp({-\frac{1}{2}}\mathbf{s}^T \mathbf{\Sigma}^{-1}\mathbf{s})}{2\Phi(\frac{t}{\sigma})-1}\cdot\mathbb{I}_{[-t,t]}(\mathbf{s})\\
&=\frac{\frac{1}{(\sqrt{2\pi\sigma^2})^{\frac{mn}{2}}}\exp({-\frac{||\mathbf{s}||^2}{2\sigma^2}})}{2\Phi(\frac{t}{\sigma})-1}\cdot\mathbb{I}_{[-t,t]}(\mathbf{s}),
\end{aligned}
\end{equation}
and
\begin{equation}
\label{eq:pdfampc2}
\begin{aligned}
&p_{\mathcal{M}_{ji}(\mathbf{S}^\prime)}(\mathbf{s})\\=&\frac{\frac{1}{(\sqrt{2\pi})^{\frac{mn}{2}}|\mathbf{\Sigma}|^{\frac{1}{2}}}\exp({-\frac{1}{2}}[\mathbf{s}-\mathrm{vec}(\mathbf{W})]^T \mathbf{\Sigma}^{-1}[\mathbf{s}-\mathrm{vec}(\mathbf{W})])}{2\Phi(\frac{t}{\sigma})-1}\\&\cdot\mathbb{I}_{[-t+\Delta,t+\Delta]}(\mathbf{s})\\
=&\frac{\frac{1}{(\sqrt{2\pi\sigma^2})^{\frac{mn}{2}}}\exp({-\frac{||\mathbf{s}-\mathrm{vec}(\mathbf{W})||^2}{2\sigma^2}})}{2\Phi(\frac{t}{\sigma})-1}\cdot\mathbb{I}_{[-t+\Delta,t+\Delta]}(\mathbf{s}),
\end{aligned}
\end{equation}
where $\mathbf{\Sigma}=\mathbb{E}[\mathbf{s}\mathbf{s}^\top]$ is the covariance matrix, for $i,j\in[N]$. 
Note that since all variables are uncorrelated, then $\mathbf{\Sigma}$ is a diagonal matrix with variances of $\sigma^2$ appearing on the main diagonal and zeros elsewhere.

\noindent
For the term $||\mathbf{s}||^2$ in \eqref{eq:pdfampc1}, we have
\begin{equation}
\label{eq:vecx}
||\mathbf{s}||^2=||\sum_{l=1}^{mn}\mathbf{s}^{[l]}||^2=\sum_{l=1}^{mn}||\mathbf{s}^{[l]}||^2.
\end{equation}
For the term $||\mathbf{s}-\mathrm{vec}(\mathbf{W})||^2$ in \eqref{eq:pdfampc2}, without loss of generality, let $\mathbf{b}_1$ be the basis that is parallel to $\mathrm{vec}(\mathbf{W})$. First, we have
\begin{equation}
\begin{aligned}
\mathbf{s}-\mathrm{vec}(\mathbf{W})=&\sum_{l=1}^{mn} \mathbf{s}^{[l]}-\mathrm{vec}(\mathbf{W}) \\ =&(\mathbf{s}^{[1]}-\mathrm{vec}(\mathbf{W}))+\sum_{l=2}^{mn} \mathbf{s}^{[i]}.
\end{aligned}
\end{equation}
Note that we have
\begin{equation}
\label{eq:iffs}
\begin{aligned}
&\mathbf{b}_1 \bot \mathbf{b}_2,...,\mathbf{b}_{mn}\\
\iff &\mathbf{s}^{[1]} \bot \mathbf{s}^{[2]},...,\mathbf{s}^{[mn]}\\
\iff &(\mathbf{s}^{[1]}-\mathrm{vec}(\mathbf{W})) \bot \mathbf{s}^{[2]},...,\mathbf{s}^{[mn]}\\
\iff &(\mathbf{s}^{[1]}-\mathrm{vec}(\mathbf{W})) \bot \mathbf{s}^{[2]}+...+\mathbf{s}^{[mn]}=\sum_{l=2}^{mn}\mathbf{s}^{[l]}.
\end{aligned}
\end{equation}
Then, since \eqref{eq:iffs} implies that $(\mathbf{s}^{[1]}-\mathrm{vec}(\mathbf{W}))$ is orthogonal to $\sum_{l=1}^{mn}\mathbf{s}^{[l]}$, by the Pythagorean theorem we have
\begin{equation}
\label{eq:vecxv}
\begin{aligned}
||\mathbf{s}-\mathrm{vec}(\mathbf{W})||^2&=||(\mathbf{s}^{[1]}-\mathrm{vec}(\mathbf{W}))||^2+||\sum_{l=2}^{mn}\mathbf{s}^{[l]}||^2\\
&\overset{\mathrm{(a)}}=(\eta_1-||\mathrm{vec}(\mathbf{W})||)^2+||\sum_{l=2}^{mn}\mathbf{s}^{[l]}||^2,
\end{aligned}
\end{equation}
where $\mathrm{(a)}$ is due to
\begin{equation}
\begin{aligned}
\lVert \mathbf{s}^{[1]}-\mathrm{vec}(\mathbf{W})\rVert^2=&
\lVert \lVert \mathbf{s}^{[1]}\rVert\cdot \mathbf{b}_1-\lVert \mathrm{vec}(\mathbf{W})\rVert\cdot \mathbf{b}_1  \rVert^2\\
=&(\lVert\mathbf{s}^{[1]}\rVert-\lVert \mathrm{vec}(\mathbf{W})\rVert)^2\\
=&(\eta_1-\lVert \mathrm{vec}(\mathbf{W})\rVert)^2.
\end{aligned}
\end{equation}
Then, by combining \eqref{eq:vecx} together with \eqref{eq:vecxv}, one can obtain the absolute value of the privacy loss function as 
\begin{equation}
\label{eq:derive}
\begin{aligned}
&\lvert l_{\mathcal{M}_{ji},\mathbf{S},\mathbf{S}^\prime}(\mathbf{s})\rvert=\lvert \mathrm{ln}(\frac{p_{\mathcal{M}_{ji}(\mathbf{S})}(\mathbf{s})}{p_{\mathcal{M}_{ji}(\mathbf{S}^\prime)}(\mathbf{s})}) \rvert\\
=&\lvert \frac{1}{2\sigma^2} (\lVert \mathbf{s}\rVert^2 -\lVert \mathbf{s}-\mathrm{vec}(\mathbf{W}) \rVert^2)\rvert\cdot\mathbb{I}_{[-t+\Delta,t]}(\mathbf{s})\\
=&\lvert \frac{1}{2\sigma^2}(\sum_{l=1}^{mn} \lVert\mathbf{s}^{[l]}\rVert^2 -(\eta_1-\lVert\mathrm{vec}(\mathbf{W})\rVert)^2+\sum_{l=2}^{mn}\lVert\mathbf{s}^{[l]}\rVert^2)\rvert\\
&\cdot\mathbb{I}_{[-t+\Delta,t]}(\mathbf{s})\\
=&\lvert \frac{1}{2\sigma^2} (\lVert \mathbf{s}^{[1]}\rVert^2 -( \eta_1-\lVert\mathrm{vec}(\mathbf{W})\rVert)^2)\rvert\cdot\mathbb{I}_{[-t+\Delta,t]}(\mathbf{s})\\
=&\lvert \frac{1}{2\sigma^2} ((\eta_1)^2 -( \eta_1-\lVert\mathrm{vec}(\mathbf{W})\rVert)^2)\rvert\cdot\mathbb{I}_{[-t+\Delta,t]}(\mathbf{s})\\
\overset{\mathrm{(b)}}\leq&\lvert \frac{1}{2\sigma^2}(-2\eta_1\Delta+(\Delta)^2) \rvert\cdot\mathbb{I}_{[-t+\Delta,t]}(\mathbf{s}),
\end{aligned}
\end{equation}
where $\mathrm{(b)}$ is based on $\lVert\mathbf{W}\rVert_F=\lVert\mathrm{vec}(\mathbf{W})\rVert\leq\Delta$, for $i,j\in[N]$. Let $y=\eta_1$ and recall that 
\begin{equation}
\label{eq:plf-truncated_scalar}
\begin{aligned}
&|l_{\mathcal{M}_{ji},d,d^\prime}(y)|=|\mathrm{ln}(\frac{p_{\mathcal{M}_{ji}(d)}(y)}{p_{\mathcal{M}_{ji}(d^\prime)}(y)})|\\
=&|\frac{1}{2\sigma^2}(-2y\Delta+(\Delta)^2)|\cdot\mathbb{I}_{[-t+\Delta,t]}(y),
\end{aligned}
\end{equation}
for $i,j\in[N]$. Based on \eqref{eq:derive} and \eqref{eq:plf-truncated_scalar}, one can write:
\begin{equation}
\label{eq:thmloss}
\begin{aligned}
\lvert l_{\mathcal{M}_{ji},\mathbf{S},\mathbf{S}^\prime}(\mathbf{s})\rvert
\leq & \lvert \frac{1}{2\sigma^2}(-2y\Delta+(\Delta)^2) \rvert\cdot\mathbb{I}_{[-t+\Delta,t]}(\mathbf{s})\\
=& [\underbrace{|l_{\mathcal{M}_{ji},d,d^\prime}(y)|,...,|l_{\mathcal{M}_{ji},d,d^\prime}(y)|}_{mn\text{'s}}],
\end{aligned}
\end{equation}
for $i,j\in[N]$. 

To ensure that the protocol is $(\epsilon,\delta)$-locally differential private, we require the absolute value of privacy loss functions as $|l_{\mathcal{M}_{ji},d,d^\prime}(y)|$'s in \eqref{eq:thmloss} to be as follows:
\begin{equation}
\label{eq:plf-trun}
\begin{aligned}
|l_{\mathcal{M}_{ji},d,d^\prime}(y)|=|\frac{1}{2\sigma^2}(-2y\Delta+(\Delta)^2)|\cdot\mathbb{I}_{[-t+\Delta,t]}(y)\leq\epsilon,
\end{aligned}
\end{equation}
for $i,j\in[N]$.
Thus, the region that cannot guarantee $(\epsilon,\delta)$-local differential privacy is 
\begin{equation*}
|\frac{1}{2\sigma^2}(-2y\Delta+(\Delta)^2)|\cdot\mathbb{I}_{[-t+\Delta,t]}(y)>\epsilon.
\end{equation*}
Therefore, we have
\begin{equation}
\label{eq:bound1}
-t\leq y<-\frac{\sigma^2\epsilon}{\Delta}+\frac{\Delta}{2},
\end{equation}
and
\begin{equation}
\label{eq:bound2}
\frac{\sigma^2\epsilon}{\Delta}+\frac{\Delta}{2}<y\leq t.
\end{equation}
From \eqref{eq:bound1} and \eqref{eq:bound2}, we must have $-t<-\frac{\sigma^2\epsilon}{\Delta}+\frac{\Delta}{2}$ and $\frac{\sigma^2\epsilon}{\Delta}+\frac{\Delta}{2}< t$. Thus, we obtain $\sigma^2<\frac{t\cdot\Delta}{\epsilon}+\frac{(\Delta)^2}{2\epsilon}$ and $\sigma^2<\frac{t\cdot\Delta}{\epsilon}-\frac{(\Delta)^2}{2\epsilon}$. Subsequently, by intersecting the inequalities and noting  that $\sigma > 0$, we obtain
\begin{equation}
\label{eq:sigmarange}
\sigma\in(0,\sqrt{\frac{t\cdot\Delta}{\epsilon}-\frac{(\Delta)^2}{2\epsilon}}).
\end{equation}
By the definition of $(\epsilon,\delta)$-local differential privacy, we know that \eqref{eq:bound1} and \eqref{eq:bound2} are the regions we cannot guarantee $(\epsilon,\delta)$-local differential privacy. 
Thus, for $Y\sim\mathcal{TN}(0,\sigma^2;[-t,t])$, we must have
\begin{equation}
\label{eq:blue}
0\leq\mathbb{P}(-t \leq Y < -\frac{\sigma^2\epsilon}{\Delta}+\frac{\Delta}{2})+\mathbb{P}(\frac{\sigma^2\epsilon}{\Delta}+\frac{\Delta}{2} < Y \leq t)\leq \delta,
\end{equation}
where
\begin{equation*}
\label{eq:LDPthm-derive}
\begin{aligned}
&\mathbb{P}(-t \leq Y < -\frac{\sigma^2\epsilon}{\Delta}+\frac{\Delta}{2})+\mathbb{P}(\frac{\sigma^2\epsilon}{\Delta}+\frac{\Delta}{2} < Y \leq t)\\
=&\frac{\Phi(\frac{-\textstyle\frac{\sigma^2\epsilon}{\Delta}+\frac{\Delta}{2}}{\sigma})-\Phi(-\frac{t}{\sigma})}{2\Phi(\frac{t}{\sigma})-1}+\frac{\Phi(\frac{t}{\sigma})-\Phi(\frac{-\textstyle\frac{\sigma^2\epsilon}{\Delta}+\frac{\Delta}{2}}{\sigma})}{2\Phi(\frac{t}{\sigma})-1}\\
=&1-\frac{\Phi(\frac{\sigma\epsilon}{\Delta}+\frac{\Delta}{2\sigma})-\Phi(-\frac{\sigma\epsilon}{\Delta}+\frac{\Delta}{2\sigma})}{2\Phi(\frac{t}{\sigma})-1},
\end{aligned}
\end{equation*}
which completes the proof.
\end{proof}

\subsection{Proof of Lemma~\ref{lemma:decrease}}
\label{app:decrease}
\begin{proof}
Let 
\begin{equation}
\label{eq:g}
g(\alpha)=\Phi(\sqrt{\frac{\epsilon}{2}}(\alpha+\frac{1}{\alpha}))-\Phi(\sqrt{\frac{\epsilon}{2}}(-\alpha+\frac{1}{\alpha})),
\end{equation} 
and 
\begin{equation}
\label{eq:h}
h(\alpha)=2\Phi(\frac{t\sqrt{2\epsilon}}{\alpha\cdot \Delta})-1. 
\end{equation}
Thus, $B(\alpha)$ is simplified as
\begin{equation}
\label{eq:simpleB}
B(\alpha)=1-\frac{g(\alpha)}{h(\alpha)},
\end{equation}
and its derivative is given by
\begin{equation}
\label{eq:Bdiff}
B^\prime(\alpha)=-\frac{g^\prime(\alpha)h(\alpha)-g(\alpha)h^\prime(\alpha)}{[h(\alpha)]^2},
\end{equation}
where
\begin{equation}
\label{eq:gdiff}
\begin{aligned}
\textstyle g^\prime(\alpha)&=\frac{\sqrt{{\epsilon}}}{2\sqrt{ \pi}} \exp \left(-\frac{\epsilon}{4}\left(\alpha+\frac{1}{\alpha}\right)^{2}\right)\cdot(1-\frac{1}{\alpha^{2}})\\
&-\frac{\sqrt{{\epsilon}}}{2\sqrt{ \pi}} \exp \left(-\frac{\epsilon}{4}\left(\alpha-\frac{1}{\alpha}\right)^{2} \right)\cdot(-1-\frac{1}{\alpha^{2}}),
\end{aligned}
\end{equation}
and
\begin{equation}
\label{eq:hdiff}
h^\prime(\alpha)=-\sqrt{\frac{2}{\pi}} \exp \left(-\frac{1}{2}\left(\frac{t \sqrt{2 \epsilon}}{\alpha \cdot\Delta}\right)^{2}\right)\left(\frac{t \sqrt{2 \epsilon}}{\alpha^{2} \cdot\Delta}\right).
\end{equation}
We start by observing the numerator in \eqref{eq:Bdiff}. In \eqref{eq:g}, since $\sqrt{\frac{\epsilon}{2}}(\alpha+\frac{1}{\alpha})>\sqrt{\frac{\epsilon}{2}}(-\alpha+\frac{1}{\alpha})$, then $\Phi(\sqrt{\frac{\epsilon}{2}}(\alpha+\frac{1}{\alpha}))>\Phi(\sqrt{\frac{\epsilon}{2}}(-\alpha+\frac{1}{\alpha}))$, resulting in 
\begin{equation}
\label{eq:gresult}
g(\alpha)>0.
\end{equation}
In \eqref{eq:h}, because $\alpha>0$, then $\frac{t\sqrt{2\epsilon}}{\alpha\cdot\Delta}>0$ and, hence, $\Phi(\frac{t\sqrt{2\epsilon}}{\alpha\cdot \Delta})>\frac{1}{2}$. Thus, one can write
\begin{equation}
\label{eq:hresult}
h(\alpha)>0.
\end{equation}
To simplify \eqref{eq:gdiff}, we have
\begin{equation}
\label{eq:gdiffresult}
\begin{aligned}
g^\prime(\alpha)=&\frac{\sqrt{{\epsilon}}}{2\sqrt{ \pi}} e^{-\frac{\epsilon}{4}(\alpha^2+\frac{1}{\alpha^2})-\frac{\epsilon}{2}}\cdot[(1-\frac{1}{\alpha^{2}})+e^{\epsilon}(1+\frac{1}{\alpha^{2}})]\\
>&\frac{\sqrt{{\epsilon}}}{2\sqrt{ \pi}} e^{-\frac{\epsilon}{4}(\alpha^2+\frac{1}{\alpha^2})-\frac{\epsilon}{2}}\cdot[(1-\frac{1}{\alpha^{2}})+e^{0}(1+\frac{1}{\alpha^{2}})]\\
=&\sqrt{\frac{\epsilon}{\pi}}e^{-\frac{\epsilon}{4}(\alpha^2+\frac{1}{\alpha^2})-\frac{1}{2}}>0.
\end{aligned}
\end{equation}
The equation \eqref{eq:hdiff} implies 
\begin{equation}
\label{eq:hdiffresult}
h^\prime(\alpha)<0.
\end{equation}
Combining \eqref{eq:gresult}, \eqref{eq:hresult}, \eqref{eq:gdiffresult}, and \eqref{eq:hdiffresult}, the numerator in \eqref{eq:Bdiff} results in $g^\prime(\alpha)h(\alpha)-g(\alpha)h^\prime(\alpha)>0$. Thus, we have
\begin{equation}
\label{eq:Bdifffrac}
\frac{g^\prime(\alpha)h(\alpha)-g(\alpha)h^\prime(\alpha)}{[h(\alpha)]^2}>0.
\end{equation}
The denominator in \eqref{eq:Bdiff} is 
\begin{equation}
\label{eq:BdiffDe}
[h(\alpha)]^2>0.
\end{equation}
Substituting \eqref{eq:Bdifffrac} into \eqref{eq:Bdiff}, yields
\begin{equation}
B^\prime(\alpha)=-\frac{g^\prime(\alpha)h(\alpha)-g(\alpha)h^\prime(\alpha)}{[h(\alpha)]^2}<0,
\end{equation}
which proves that $B(\alpha)$ is a monotonic decreasing function for $\alpha \in (0,\sqrt{\frac{2t}{\Delta}-1}]$.
\end{proof}

\subsection{Proof of Theorem~\ref{thm:submatrices}}
\label{app:submatrices}
\begin{proof}
Recall that $\lbrack \mathbf{X}_j\rbrack_i=\mathbf{X}_j+\tilde{\mathbf{N}}_{ji}$, where  $\tilde{\mathbf{N}}_{ji}=\sum_{k=1}^T \omega_i^k\mathbf{N}_{j,k}$ with all entries in $\mathbf{N}_{j,k}$'s are randomly distributed according to  $\mathcal{N}(0,\sigma^2_s)$, for $i,j\in[N]$. We resample the entire $\tilde{\mathbf{N}}_{ji}=\sum_{k=1}^T \omega_i^k\mathbf{N}_{j,k}$ by randomly generating $\mathbf{N}_{j,k}$'s until the result is within the range of $[-t,t]$, for $i,j\in[N]$. By the property of linear combination of Gaussian distributions, the combined noises at each entry are in a distribution of $\mathcal{N}(0,\sum_{k=1}^T \lvert\omega_i^k\rvert^2\sigma_s^2)$, for $i\in[N]$. Since we truncate the distribution to $[-t,t]$, which obtains a distribution of $\mathcal{TN}(0,\sum_{k=1}^T \lvert\omega_i^k\rvert^2\sigma_s^2;[-t,t])$, for $i\in[N]$. We have
\begin{equation}
\sum_{k=1}^T \lvert\omega_i^k\rvert^2\sigma_s^2=(\sum_{k=1}^T \lvert\omega_i^k\rvert^2)\cdot\sigma_s^2=\sigma^2,
\end{equation}
where $\sigma=\frac{\alpha^*\cdot\Delta}{\sqrt{2\epsilon}}$, for $i\in[N]$. Therefore, one can characterize  the variance of the noise generated for each entry in the noise matrices $\mathbf{N}_{j,1},\dots,\mathbf{N}_{j,T}$ as 
\begin{equation}
\sigma_s^2=\frac{\sigma^2}{\sum_{k=1}^T \lvert\omega_i^k\rvert^2}=\frac{(\frac{\alpha^*\cdot\Delta}{\sqrt{2\epsilon}})^2}{T}=\frac{(\alpha^*)^2 \cdot \Delta^2}{2\epsilon T},
\end{equation}
where $\lvert\omega_i^k\rvert=1$, for $i,j\in[N]$.
\end{proof}

\subsection{Description of Algorithm~\ref{algo:privatemul}}
\label{app:privatemul}
In Algorithm \ref{algo:privatemul}, Step $1$ to Step $9$ compute the secret shares of $[[\mathbf{U}]_i\times[\mathbf{V}]_j]_k$ for client $k$, for $i,j,k\in[N]$. 
In order to obtain $[\mathbf{UV}]_k$, client $k$ locally computes
\begin{equation}
\begin{aligned}
&\frac{1}{N^2}\sum_{i=1}^{N}\sum_{j=1}^{N}[[\mathbf{U}]_i\times[\mathbf{V}]_j]_k
=
[\frac{1}{N^2}\sum_{i=1}^{N}\sum_{j=1}^{N}[\mathbf{U}]_i\times[\mathbf{V}]_j]_k\\
= & 
[\frac{1}{N}\sum_{j=1}^{N}(\frac{1}{N}\sum_{i=1}^{N}[\mathbf{U}]_i)\times[\mathbf{V}]_j]_k
= 
[\frac{1}{N}\sum_{j=1}^{N}\mathbf{U}\times[\mathbf{V}]_j]_k\\
= &
[\mathbf{U}\times\frac{1}{N}\sum_{j=1}^{N}[\mathbf{V}]_j]_k
=[\mathbf{U}\mathbf{V}]_k,
\end{aligned}
\end{equation}
where $\frac{1}{N}\sum_{i=1}^{N}[\mathbf{U}]_i=\mathbf{U}$ and $\frac{1}{N}\sum_{j=1}^{N}[\mathbf{V}]_j=\mathbf{V}$, for $k\in[N]$. 
Thus, at the end of Algorithm \ref{algo:privatemul}, each client holds a secret share of $\mathbf{U}\mathbf{V}$ as $[\mathbf{U}\mathbf{V}]_k$, for $k\in[N]$.

\bibliographystyle{IEEEtran}
\bibliography{Paper} 

\end{document}